\documentclass[letterpaper,USenglish,cleveref,autoref,thm-restate,numberwithinsect]{lipics-v2021}
\pdfoutput=1 %uncomment to ensure pdflatex processing (mandatatory e.g. to submit to arXiv)
\title{Adaptive Sampling for Minimum-Norm $k$-Clustering} %TODO Please add

\author{Haripriya Pulyassary}{Cornell University, Ithaca, NY
  14850}{hp297@cornell.edu}{https://orcid.org/0000-0002-4100-0678}%
  {Part of this work was done while the author was a Master's student at the University of Waterloo.}
\author{Chaitanya Swamy}{Dept. of Combinatorics and Optimization, Univ. Waterloo,
  waterloo, ON N2L 3G1}{cswamy@uwaterloo.ca}{https://orcid.org/0000-0003-1108-7941}%
  {}

\authorrunning{H. Pulyassary and C. Swamy} 

\Copyright{Harpriya Pulyassary and Chaitanya Swamy} 



\category{} %optional, e.g. invited paper

\relatedversion{A preliminary version is to appear in the Proceedings of the
  European Symposium on Algorithms (ESA) 2026.} %optional, e.g. full version hosted on arXiv, HAL, or other respository/website
\funding{Supported in part by C. Swamy's NSERC Discovery grant.}

\nolinenumbers %uncomment to disable line numbering

\hideLIPIcs
\EventEditors{Philip Bille, Seth Pettie, and Sabine Storandt}
\EventNoEds{3}
\EventLongTitle{34th Annual European Symposium on Algorithms (ESA 2026)}
\EventShortTitle{ESA 2026}
\EventAcronym{ESA}
\EventYear{2026}
\EventDate{August 31--September 4, 2026}
\EventLocation{L'Aquila, Italy}
\EventLogo{}
\SeriesVolume{388}
\ArticleNo{147}
\usepackage{algorithm,algorithmic}
\usepackage{amsmath,amsthm, amssymb}
\usepackage{comment}
\usepackage{graphicx} % Required for inserting images
\usepackage{placeins}
\usepackage{xcolor}
\usepackage{xspace}
\newenvironment{proofof}[1]{\begin{proof}[Proof of {#1}]}{\end{proof}}

\DeclareMathOperator{\poly}{poly}

\DeclareMathOperator{\prev}{prev}
\DeclareMathOperator{\core}{core}
\DeclareMathOperator{\supcore}{supercore}

\newcommand{\vt}{\vec{t}}
\newcommand{\bad}{\ensuremath{\mathsf{bad}}}

\newcommand{\clustroot}{C^*}
\newcommand{\clust}{\clustroot_q}
\newcommand{\clustq}{\clustroot_q}
\newcommand{\cli}{j}

\newcommand{\lcore}[1][\ell]{\ensuremath{\core_{#1}}}

\newcommand{\optcen}{o^*}
\newcommand{\optSoln}{O^*}
\newcommand{\optsoln}{\optSoln}
\newcommand{\pos}{\ensuremath{\mathsf{POS}}}
\newcommand{\post}{\pos}
\newcommand{\rl}{r_\ell}

\newcommand{\topl}[1][\ell]{\ensuremath{\text{Top}_{#1}}}
\newcommand{\veps}{\varepsilon}

\newcommand{\costly}{costly\xspace}

\DeclareMathOperator{\Exp}{\mathbb{E}}

\newcommand{\E}[2][{}]{\ensuremath{{\textstyle{\Exp}_{#1}}\bigl[#2\bigr]}}

\newcommand{\R}{\ensuremath{\mathbb R}}
\newcommand{\Rp}{\R_{\geq 0}}

\newcommand{\C}{\ensuremath{\mathcal{C}}}
\newcommand{\F}{\ensuremath{\mathcal{F}}}

\newcommand{\T}{\ensuremath{\mathcal T}}

\newcommand{\OPT}{\ensuremath{\mathit{OPT}}}
\newcommand{\opt}{\OPT}
\newcommand{\cost}{\ensuremath{\mathit{cost}}}

\newcommand{\es}{\ensuremath{\varnothing}}
\newcommand{\assign}{\ensuremath{\leftarrow}}

\newcommand{\ceil}[1]{\ensuremath{\bigl\lceil#1\bigr\rceil}}

\newcommand{\e}{\ensuremath{\epsilon}}

\newcommand{\gm}{\ensuremath{\gamma}}

\newcommand{\sse}{\subseteq}
\newcommand{\ld}{\ensuremath{\lambda}}
\newcommand{\kp}{\ensuremath{\kappa}}
\newcommand{\al}{\ensuremath{\alpha}}
\newcommand{\tht}{\ensuremath{\theta}}

\newcommand{\dt}{\ensuremath{\delta}}

\newcommand{\ve}{\ensuremath{\varepsilon}}

\newcommand{\down}{\ensuremath{\downarrow}}
\newcommand{\indset}{\ensuremath{I}}
\newcommand{\lb}{\ensuremath{\mathsf{lb}}}
\newcommand{\ub}{\ensuremath{\mathsf{ub}}}
\newcommand{\tauval}{35}
\newcommand{\betaval}{3}
\newcommand{\rhoval}{35}
\newcommand{\csalg}{\ensuremath{\mathsf{CSAlg}}\xspace}
\newcommand{\tS}{\ensuremath{\widetilde S}}
\newcommand{\lptime}{\ensuremath{\mathsf{LPsolve}}\xspace}
\newcommand{\hit}{H}
\newcommand{\unhit}{U}
\newcommand{\lost}{W}
\newcommand{\hitcost}{\Lambda}
\newcommand{\num}{\ensuremath{\mathsf{num}}}
\newcommand{\nbetaval}{3}

\newcommand{\sparse}{sparsified\xspace}
\newcommand{\htt}{\ensuremath{\widehat{t}}}

\makeatletter
\def\@citex[#1]#2{\leavevmode
  \let\@citea\@empty
  \@cite{\@for\@citeb:=#2\do
    {\@citea\def\@citea{,\penalty\@m\ }%
\edef\magic##1{\let##1\expandafter\noexpand\csname bibalias@\@citeb\endcsname}%
\magic\tmp \ifx\tmp\relax\else \let\@citeb\tmp\fi
     \edef\@citeb{\expandafter\@firstofone\@citeb\@empty}%
     \if@filesw\immediate\write\@auxout{\string\citation{\@citeb}}\fi
     \@ifundefined{b@\@citeb}{\hbox{\reset@font\bfseries ?}%
       \G@refundefinedtrue
       \@latex@warning
         {Citation `\@citeb' on page \thepage \space undefined}}%
       {\@cite@ofmt{\csname b@\@citeb\endcsname}}}}{#1}}
\def\bibalias#1#2{\expandafter\def\csname bibalias@#1\endcsname{#2}}
\makeatother

\begin{document}

\bibalias{ArthurV07}{arthur_k-means_2007}

\maketitle

\begin{abstract}
    In $k$-clustering problems, we are given a metric space $(\C, d)$, and must choose a set $S$ of $k$ centers to open. Each client $j \in \C$ incurs an assignment cost, which is the distance between $j$ and center in $S$ that it has been assigned to. In this work, we study the \emph{minimum-norm $k$-clustering problem}, where we are given an arbitrary monotone symmetric norm $f$, and wish to open $k$ centers so as to minimize $f$(assignment-cost vector). This is a powerful generalization, encompassing many classical $k$-clustering problems including the $k$-median, $k$-means, and $k$-center problems. 

    A simple and efficient algorithmic idea  is that of \emph{adaptive sampling}, wherein
    we randomly choose the location of the next center to open with probability
    proportional to its ``cost" under the currently chosen set. While this has yielded
    fast algorithms for some $k$-clustering problem, little is known for settings
    \emph{without} ``min-sum" objectives. 

We devise the first adaptive-sampling-based bicriteria constant-factor approximation
algorithm for general minimum-norm $k$-clustering, vastly expanding the scope of problems
handled by adaptive sampling. For the special case of $\topl$ norms, which form a building
block of monotone symmetric norms, we show that adaptive sampling yields an 
$O(\log k)$-approximation algorithm. 
\end{abstract}
%\priya{I'll add my changes in blue so that it is easier to track}

\section{Introduction}
Clustering problems are ubiquitous, and arise in a variety of application domains 
%ranging from
including data mining, machine learning, computer vision, computational geometry, 
bioinformatics, and supply-chain logistics. 
Typically, in a $k$-clustering problem, we are given a
metric space $(\C, d)$ and must choose a set $S$ of $k$ centers to open. Each client 
$j\in \C$ incurs a connection or assignment cost, namely the distance between $j$ and center
in $S$ that it has been assigned to (which is often the closest center in $S$). 
%Depending on the application, one typically 
In many such problems, the goal is to choose
$k$-centers so as to minimize some function of the induced assignment-cost
vector. Classical problems of this form, which have been extensively studied, include the
{\em $k$-median}, {\em $k$-means}, and {\em $k$-center} problems, 
wherein the function (to be minimized) is the \emph{sum}, sum of squares,
and {\em maximum}, of the assignment-cost vector entries respectively.
%{\em $k$-center}, where the function is the \emph{maximum} of the induced assignment costs
%(i.e. the $k$-median and $k$-center 
%problems respectively).  
%Another well-studied problem of this form is the {\em $k$-means} problem, where the
%objective is to minimize the sum of squares of the assignment-cost vector entries.

In this paper, we study the \emph{minimum-norm $k$-clustering problem}, which is a
far-reaching generalization of the $k$-median and $k$-center problems (and also
essentially captures the $k$-means problem).
%both of these objectives. 
In this problem, we are given an {\em arbitrary monotone, symmetric norm}
$f:\R^n\mapsto\R_+$, where $n=|\C|$, and we wish to minimize 
$f(\text{assignment-cost vector})$. Monotone, symmetric norms capture a very wide range of 
objective functions; we list two relevant examples here, and refer the reader
to~\cite{ChakrabartyS19} and the references therein for further examples.

\begin{itemize}
\item \textbf{$\ell_p$-norms}: 
  These are probably the most well-known example of monotone, symmetric norms, where
  %A classical example of a monotone symmetric norm is the $\ell_p$-norm, 
  $\ell_p(v) := (\sum_{j=1}^n |v_j|^p)^{1/p}$, for $p\geq 1$. When $p = 1, \infty$, we
  obtain the $k$-median and $k$-center problems respectively, and $p=2$ essentially
  captures the $k$-means problem (the $k$-means objective is $\ell_2^2$).
 
  \item \textbf{$\topl$ and ordered norms}: The {\em $\topl$-norm} of a vector $v\in\R^n$
    is the sum of the $\ell$ largest entries of $(|v_j|)_{j\in[n]}$. 
    Letting $v^\downarrow$ denote the vector $v$ with its entries sorted in non-increasing
    order, for $v\geq 0$, we have $\topl(v) = \sum_{j = 1}^\ell v^\downarrow_j$. 
    An {\em ordered norm} is a nonnegative linear combination of $\topl$ norms;
    %A generalization of the $\topl$ norm is the ordered norm, where,
    equivalently, we can describe this by a non-increasing weight vector $w \ge 0$, which 
    %$v$ is $cost(v; w) =
    defines the ordered norm $f(v):=\sum_{j=1}^n w_j v^\downarrow_j$ for $v\geq 0$. 
    The problems of minimizing the $\topl$-norm and ordered-norm of the induced cost
    vector are referred to as the \emph{$\ell$-centrum} and \emph{ordered $k$-median}
    problems respectively in the clustering literature~\cite{tamir2001k,aouad_ordered_2018,byrka_constant-factor_2018,chakrabarty_interpolating_2018,ChakrabartyS19}.
\end{itemize}

Furthermore, the class of monotone symmetric norms is closed under nonnegative linear
combinations, and taking the maximum (over a finite collection of monotone symmetric
norms). Also, for any monotone norm $g$, $f = g(v^\downarrow)$ and $f =
\mathbb{E}_\pi[g(\{v_{\pi(i)}\}_{i \in [n]})]$ (where $\pi$ is a random permutation) are
monotone symmetric norms. Thus, minimum-norm $k$-clustering is indeed a very versatile
problem that can be used to model a wide range of applications.  

An elegant algorithmic idea that has been applied to some $k$-clustering problems is
\emph{adaptive sampling} \cite{arthur_k-means_2007,aggarwal_adaptive_2009} (see
also~\cite{OstrovskyRSS06}). In adaptive sampling, we randomly choose the location of the
next center to open with probability proportional to the objective-function contribution
of the point under the currently chosen set. 
The power of adaptive sampling lies in its  simplicity; unlike
other more complicated (and often linear-programming-based) algorithms, adaptive sampling
is fast, and uses pairwise-distance information in a conservative manner. This makes
adaptive-sampling-based algorithms well-suited for computing a preliminary set of centers
in data-sparsification procedures (e.g. coreset construction
\cite{chen_coresets_2009,braverman_coresets_2019}),  as well as other situations where
pairwise-distance queries are expensive
(e.g. \cite{burkhardt_low-distortion_2024,aaai25,aaai25b}).  

It is known that adaptive sampling can be used to obtain a bicriteria
constant-factor approximation for the $k$-median and $k$-means clustering problems  
\cite{aggarwal_adaptive_2009}; by this we mean, %we obtain 
a solution that opens at most $\al k$ centers and has cost at most $\beta$ times the
optimum---denoted an $(\al,\beta)$-approximate solution---where $\al,\beta=O(1)$. 
For both these problems, as noted earlier, the objective is a
min-sum objective, where we sum up
the costs incurred by the individual clients. 
%(where the cost of a client $j$ the distance, or squared distance, between $j$ and the
%open centers). 
Consequently, sampling
the next center to open with probability proportional to its current cost does indeed
yield good solutions (with constant probability). 
%\footnote{In the context of $k$-means, adaptive sampling is also called $D^2$-sampling
%referring to the fact that we sample proportional to the squared distances.}
While adaptive sampling, which is also called {\em $D^2$-sampling} in
the context of the $k$-means objective, has been leveraged in various settings---e.g., for
other objectives~\cite{ackermann_coresets_2009,aggarwal_adaptive_2009,statman2020k},
constrained $k$-means problems~\cite{BhattacharyaGJK20,BhattacharyaJK18,BhaskaraVX19}, 
streaming settings~\cite{ailon2009streaming,BravermanORST11}, models where other clustering
information is available~\cite{AilonBJK18}---%
to our knowledge, with the exception of the very recent work~\cite{aaai25,aaai25b}, 
%(which we discuss below), 
all these applications of adaptive sampling have involved
%of adaptive
%sampling~\cite{arthur_k-means_2007,aggarwal_adaptive_2009,ackerman_coresets_2009} 
$k$-clustering with {\em min-sum objectives}, a milieu naturally suited for adaptive sampling. 
%is naturally suited for. 
%(albeit it has been utilized to
%handle some constrained $k$-clustering settings~\cite{BhattacharyaJK18,BhaskaraVX19}
 
Recently, \cite{aaai25,aaai25b} extended the adaptive-sampling
approach to obtain a bicriteria constant-factor approximation for the $\ell$-centrum
problem (i.e., where the norm is a $\topl$ norm).
A notable aspect of this result %that is worth pointing out 
is that,
unlike a min-sum objective, the $\topl$-norm objective is not {\em separable}, i.e., we
cannot quite isolate the contribution of a client to the objective, and consequently, it
is not clear how to utilize the adaptive-sampling template here.
%described earlier, where we sample a center at a point proportional 
Nevertheless, \cite{aaai25,aaai25b} observed that one can sidestep this issue by
%exploiting the fact that the $\topl$-norm is only mildly non-separable: 
%one can well-estimate %the $\topl$-norm is 
%it can be well-estimated by a 
moving to the min-sum (hence, separable) proxy function introduced
by~\cite{ChakrabartyS19}, which well-estimates the $\topl$-norm.
%which has since become a standard tool in dealing with $\topl$ norms.
In this sense, $\topl$-norms are only mildly non-separable, and by using adaptive 
sampling on this separable proxy function (i.e., sampling a client with probability
proportional to its proxy-cost contribution), 
one can obtain an $\bigl(O(1),O(1)\bigr)$-approximate solution for 
%the $\ell$-centrum problem
$\topl$ norms~\cite{aaai25,aaai25b}. 

%observing that instead
%of working with the $\topl$-objective directly, one can move to the proxy function for the
%$\topl$-objective introduced by~\cite{ChakrabartyS19}, which is a min-sum

\subsection{Our contributions}
%In this paper, 
{\em We design an adaptive-sampling-based bicriteria constant-factor approximation
  algorithm for general minimum-norm $k$-clustering.}
Given the generality of minimum-norm $k$-clustering, we thus obtain a fast, simple
algorithm that is versatile and can be applied to a very wide range of $k$-clustering
problems. 
The simplicity and efficiency of adaptive sampling also opens up the possibility 
of leveraging our ideas to handle minimum-norm $k$-clustering in 
%other settings, such as 
the streaming setting (as has been done for $k$-means).

The only other polynomial-time guarantee known for minimum-norm $k$-clustering is a
constant-factor  
approximation algorithm due to Chakrabarty and Swamy~\cite{ChakrabartyS19}. 
%In comparison with our result, 
They obtain a ``true'' approximation, %(with a large approximation factor),
i.e., the solution returned opens (at most) $k$ centers, 
%while we obtain a bicriteria approximation, 
but their algorithm is much more involved: it is based on solving a suitable
LP-relaxation of the problem, which makes it less efficient than our adaptive-sampling 
algorithm. It is worth noting that %since we obtain an
%$\bigl(O(1),O(1)\bigr)$-approximate solution, one can use it 
%to sparsify the data set to $O(k)$ points, and then 
by using our algorithm to sparsify the data, and then running the algorithm
in~\cite{ChakrabartyS19}, we can obtain a (true) $O(1)$-approximation algorithm with
better running time %(than the algorithm in~\cite{ChakrabartyS19}) 
that succeeds with
constant probability; see Section~\ref{trueapx}. 

We also show that for $\topl$ norms, which form a building block of monotone symmetric
norms, {\em adaptive sampling yields a (true) $O(\log k)$-approximation}. Previously, such
guarantees were known only for min-sum objectives when the underlying distance function
$d$ (determining the client assignment costs) satisfies the approximate triangle
inequality $d(p,q)\leq\rho\bigl(d(p,r)+d(r,q)\bigr)$ for all points $p,q,r$, for some
parameter
$\rho$~\cite{arthur_k-means_2007,ackermann_coresets_2009,statman2020k}. Notably, the  
proxy function that renders $\topl$-norm as a min-sum objective does not satisfy this
approximate triangle inequality for {\em any} finite $\rho$, so we need to supplement
the analysis in~\cite{arthur_k-means_2007} with new ideas.

\subparagraph*{Technical overview.}
Technically, the highly non-separable nature of general monotone, symmetric norms 
(for instance, the norm could be the maximum of any collection of monotone, symmetric
norms) 
%is itself a monotone symmetric norm) 
presents a significant challenge, and makes it much more challenging to deal with 
this general class of objectives as compared to $\topl$ norms.
%While $\topl$-norms are mildly non-separable, %(as discussed above), monotone
%symmetric norms can be highly non-separable (for instance, the maximum of
%monotone, symmetric norms is itself a monotone symmetric norm), and
In particular, {\em unlike the case of $\topl$-norms, there is no simple proxy-cost
function that one can switch to, that renders the minimum-norm problem separable}. 
%
%Indeed, this utter lack of
%separability, %of monotone, symmetric norms, 
%We point out that 
%makes dealing with an arbitrary monotone, symmetric norm significantly more
%challenging than dealing with $\topl$-norms.  
%
As mentioned above, %discussed in Section~\ref{adsample-topl}, 
having such a separable proxy function was {\em key} towards extending
adaptive sampling to make it work for $\topl$-norms in~\cite{aaai25,aaai25b}. 
%since we could use the per-client contribution towards the separable proxy cost to bias
%the sampling process appropriately.
However, without access to such a separable function for a general monotone, symmetric
norm, we need to come up with new ideas. 
Indeed, this highly non-separable nature of general monotone, symmetric norms, which 
manifests itself in terms of the lack of a convenient separable proxy cost function,  
%makes it significantly more challenging to deal with
%this general class of objectives as compared to $\topl$ norms.
often presents an obstacle in extending results obtained for $\topl$-norms to general
monotone, symmetric norms, and for a variety of combinatorial-optimization problems
(e.g., $k$-clustering, shortest path, matching), we have much less of an understanding of 
the minimum-norm variant of the problem compared to the $\topl$-norm variant of the
problem.% 
\footnote{As an illustrative example, consider {\em minimum-norm $s$-$t$ path}, where
we seek an $s$-$t$ path $P$ to minimize $f(\vec{c}^P)$, where $\vec{c}^P\in\R^E$ is the
vector whose entry for edge $e$ is $c_e$ if $e\in P$ and $0$ otherwise.
When $f$ is a $\topl$-norm, the min-sum proxy function of~\cite{ChakrabartyS19} conveniently
reduces the problem to %a min-sum problem, i.e., 
a standard $s$-$t$ shortest-path problem, which can be solved efficiently. 
%which is often well-understood.
%so all we need to do is tackle the min-sum problem. 
To solve the problem with a general monotone, symmetric norm, %as discussed below, 
%one way of solving the minimum-norm variant is by 
%to solve the problem, 
we need to control multiple $\topl$-norms; this translates to an
$s$-$t$ path problem with multiple knapsack or cardinality constraints; %(see~\cite{ChenLRZ25}), 
%(depending on how these $\topl$-norm constraints are encoded), 
this changes the nature of the problem significantly, and only an 
$O(\log\log n)$-approximation algorithm is known~\cite{ChenLRZ25}. 
%The resulting problem is related to the {\em robust $s$-$t$ path}
%problem~\cite{KasperskiZ09}. 
%which is hard to approximate within a factor better
%than $\Omega\bigl(\log(\text{no. of knapsack constraints})\bigr)$~\cite{KasperskiZ09}.
}  

%When considering general monotone symmetric norms, we no longer have a simple separable
%proxy function. 
The way forward is to utilize majorization theory~\cite{HardyLP}, which tells us that 
we can control the norm $f(v)$ by controlling all $\topl$-norms of $v$. More precisely
(see Theorem~\ref{thm:majorization}),
to ensure that $f(v)\leq O(1)\cdot f(o^*)$, where $o^*$ is the cost-vector
induced by an optimal solution, it suffices to ensure that 
$\topl(v)\leq O(1)\cdot\topl(o^*)$ for all $\ell$, or, in fact, 
%all $\ell$s that are 
(roughly speaking) geometrically increasing values of $\ell$.
%However, we are able to utilize the characterization via majorization: for
%any monotone symmetric norm $f$, to control the cost with respect to $f$, it suffices to
%control the $\topl$-cost of all $\ell$ (or a logarithmically-sized subset). 
Since we have an adaptive-sampling algorithm for $\topl$ norms~\cite{aaai25,aaai25b}, this
seems promising; however, the execution of this algorithm, in particular, the sampling
process, (naturally) depends on the index $\ell$, whereas we need to control the
$\topl$-norms for multiple values of $\ell$ {\em simultaneously}. Also, note that simply
running the $\topl$-norm algorithm multiple times 
will not work, since we will end up opening far too many centers, $O(k\log n)$.
%$O(k \log n)$ centers.

Nonetheless, we are able to devise an adaptive-sampling algorithm
(Algorithm~\ref{adsample-alg}) for this problem by
exploiting some key insights from the {\em analysis} of the adaptive-sampling algorithm
for $\topl$ norms in~\cite{aaai25,aaai25b}.
%the following key structural observation. 
The analysis therein (see Section~\ref{adsample-topl}) shows that the 
%$\ell$-centrum adaptive sampling 
algorithm makes progress when it chooses a center in the  ``core'' of a cluster induced by
an optimal solution that is currently costly (Claim~\ref{coreisgood}), and that this event
happens with constant probability as long as the current $\topl$-cost is large
(Lemma~\ref{choosecore}). 
While these cores differ for different values of $\ell$, for any fixed cluster, 
they form a nested family: the cores get smaller as $\ell$ increases
(Claim~\ref{corenested}). %laminar family.  
This yields the following extremely-fruitful insight, which drives our algorithm:
if at each step, we concentrate 
%only on controlling the $\topl$-cost for 
on the \emph{largest} index $\ell$ for which the current $\topl$-cost is large, 
%uncontrolled, the
then with constant probability, the new center that we open {\em will simultaneously make
progress with respect to all costly $\topl$-norms}. 
Hence, using a standard martingale argument, we can show that after $O(k)$ centers
are opened, the $\topl$-cost of our solution is bounded for all indices $\ell$ (with
constant probability), and this implies that we have an
$\bigl(O(1),O(1)\bigr)$-approximate solution (Theorem~\ref{adsample-thm}). 
%is controlled 
 
\subsection{Related literature} \label{relwork}
The $\ell$-centrum and its generalization, the ordered $k$-median problem, have been
extensively studied in the Operations Research literature (see, e.g.,
\cite{aouad_ordered_2018, laporte2019introduction,nickel2006location} and the references
within). The first $O(\log n)$-approximation algorithms for $\ell$-centrum and the ordered
$k$-median problems were given by Tamir \cite{tamir2001k} and Aouad and Segev
\cite{aouad_ordered_2018} respectively. Subsequently, %Byrka, Sornat, and Spoerhase
Byrka et al.~\cite{byrka_constant-factor_2018}, and Chakrabarty and Swamy
\cite{chakrabarty_interpolating_2018} gave the first constant-factor approximations for
these two problems. %Alamdari and Shmoys paper needed? 
%These problems are further generalized by the \emph{minimum-norm $k$-clustering problem}. 
Chakrabarty and Swamy~\cite{ChakrabartyS19} introduced the much-more general minimum-norm
$k$-clustering problem, and used LP techniques to devise a $(408+\ve)$-approximation
algorithm for the problem that opens exactly $k$ centers.
To the best of our knowledge, this the only polytime approximation result known for
monotone symmetric norms. 
But there has been some work on parameterized approximation algorithms for
even-more general norm-clustering problems \cite{abbasi2023parameterized}. 
%is a constant-factor approximation algorithm due to Chakrabarty and Swamy
%\cite{ChakrabartyS19}. 
%It should be noted that the minimum-norm $k$-clustering problem generalizes many
%well-studied $k$-clustering problems, including $\ell$-centrum, ordered $k$-median, and
%$k$-means clustering.   
 
In this paper, we study adaptive sampling algorithms for the minimum-norm $k$-clustering problem. Adaptive sampling (initially coined as $D^2$ sampling in the context of the $k$-means problem), was introduced and studied by Arthur and Vassilvitskii \cite{arthur_k-means_2007}. Their adaptive sampling-based algorithm, \verb|k-means++|, yields a solution of expected cost $O(\log k) \cdot \opt$; this bound was later proved to be tight, even in two dimensions \cite{bhattacharya_tight_2016}. With some modifications to the original \verb|k-means++| algorithm, it is possible to obtain improved runtime and approximation guarantees. In particular, Bachem et al. \cite{bachem_approximate_2016} showed that one can obtain an $O(\log k)$-approx in \emph{sublinear time} by replacing the exact $D^2$-sampling step with a (faster) approximation based on Markov Chain Monte Carlo sampling. Aggarwal, Deshpande, and Kannan \cite{aggarwal_adaptive_2009} later proved that using the $D^2$-sampling protocol to open $O(k)$ centers (instead of exactly $k$ centers) yields a constant factor bicriteria approximation for $k$-means. %; moreover, sampling with probability proportional to the distance from the currently open centers, yields a $(O(1), O(1))$-bicriteria approximation for the $k$-median problem as well. 

Adaptive sampling has been used to give approximation algorithms for other variants of the
$k$-means problem as well, including under the streaming model
\cite{ailon2009streaming,BravermanORST11}, clustering with outliers
\cite{BhaskaraVX19,deshpande2020robust}, other constrained $k$-means problems
\cite{BhattacharyaGJK20,BhattacharyaJK18}, %This technique 
%It has also been utilized by approximation algorithms for 
as also other objectives such as minimizing the sum of the $p$th-powers of
distances~\cite{aggarwal_adaptive_2009}, 
Bregman divergences~\cite{ackermann_coresets_2009}, truncated costs~\cite{statman2020k},
and $\topl$-norms~\cite{aaai25,aaai25b}. 
As noted earlier, except for~\cite{aaai25,aaai25b}, all these works consider min-sum, and
hence, separable objectives, and~\cite{aaai25,aaai25b} also move to a min-sum proxy
function that well-estimates the $\topl$ norm. 

\section{Preliminaries} \label{prelim}

\subparagraph*{Monotone, symmetric norms.}
Recall that a function  $f : \mathbb{R}^n \to \mathbb{R}_+$ is a norm if it satisfies (i)
$f(x) = 0$ iff $x = 0$; (ii) $f(x + y) \leq  f(x) + f(y)$ for all $x, y \in \mathbb{R}^n$
(triangle inequality); and (iii) $f(\lambda x) = |\lambda| f(x)$ for all $x \in
\mathbb{R}^n$, $\lambda \in \mathbb{R}$ (homogeneity). 
If $f$ is \emph{monotone}, then $f(x) \le f(y)$ for all $0 \le x \le y$,  and
$f$ is \emph{symmetric} if permuting the coordinates of $v$ does not change the value of
$f(v)$.    

We collect here various properties of monotone, symmetric norms from prior work that we
will utilize. 
For a vector $v\in\R^n_+$, we use $v^\downarrow$ to denote $v$ with its coordinates sorted in
non-increasing order. 
%So $f$ being symmetric means that $f(v)=f(v^\downarrow)$. 
A special class of monotone symmetric norms which will be particularly useful are $\topl$
norms. 

\begin{definition}  
  The $\topl$-norm of $v \in \mathbb{R}^n_+$, denoted  $\topl(v)$, is the
  sum of the $\ell$ largest coordinates of $v$, i.e., $\topl(v) = \sum_{i=1}^\ell
  v^\downarrow_i$. 
\end{definition}

$\topl$-norms can be difficult to work with because they are non-separable, i.e., it is
unclear how to isolate the contribution from a single coordinate $v_i$ towards $f(v)$.
But this difficulty can be somewhat mitigated by working with a separable
proxy function---the expression on the RHS of Claim~\ref{csproxy} (a)---introduced
by~\cite{ChakrabartyS19}.  
(This proxy function does not however 
satisfy the triangle inequality, which does cause other complications.)
For $z\in\R$, define $z^+:=\max\{z, 0\}$.    

\begin{claim}[Claims 6.1 and 6.2 in~\cite{ChakrabartyS19}] 
\label{claim:CS19proxy} \label{csproxy}		 
Let $v \in \mathbb{R}_{+}^n$ and $\rho \in  \mathbb{R}_{+}$. Then,
\begin{enumerate}[(a)] %[label=(\alph*), topsep=0.2ex, itemsep=0.1ex, leftmargin=*]
\item $\topl(v) \leq \ell \cdot \rho + \sum_{i=1}^n (v_i - \rho)^+$; %and
\item  if $v^\downarrow_\ell \leq \rho \leq (1+\varepsilon)v^\downarrow_\ell$, we have
$\ell \cdot \rho + \sum_{i=1}^n (v_i - \rho)^+ \leq (1+\varepsilon) \cdot \topl(v)$. 	 
\end{enumerate} 
\end{claim}

An important property of $\topl$-norms is that they serve as a handle for reasoning about
arbitrary monotone symmetric norms. Indeed, a key insight %/ideas
that is useful when dealing with monotone symmetric norms, and leveraged by our
algorithm as well, 
%as also other  including by our algorithm, by our algorithm, as also in Section
%\ref{sec:msn} 
is that for any monotone symmetric norm $f$, in order to control $f(v)$, it suffices to
control $\topl(v)$ for all $\ell \in [n]$.  
In fact, it suffices to control only logarithmically-many $\topl$ norms 
(incurring a small loss in approximation quality). 
%is sufficient to control the induced $f$-cost. 
%To be precise, we 

\begin{definition}\label{posdef}
Define %the index-set 
$\pos_{n, \delta}\sse[n]$ as follows. We include index $1$ in
$\pos_{n,\dt}$. Subsequently, if $\ell$ is the current-largest index in $\pos_{n,\dt}$ and 
$\ceil{(1+\dt)\ell}\leq n$, we include $\ceil{(1+\dt)\ell}$ in $\pos_{n,\dt}$.
We abbreviate $\pos_{n,\dt}$ to $\pos$, when $n,\dt$ are clear from the context. 
\end{definition}

We are using the definition of $\pos_{n,\dt}$ in~\cite{IbrahimpurS21}, which is slightly
different from the definition in~\cite{ChakrabartyS19},  
as we will utilize certain results from~\cite{IbrahimpurS21};
but this difference is not crucial. 
%we are using the definition in~\cite{IbrahimpurS21}  

%:= \{\min\{\lceil(1 + \delta)^s\rceil, n\} : s \geq 0\}$; 
%For any $r\in[n]$, let $\prev(r)$ be the largest index $\ell\leq r$ in $\pos$. 

\begin{theorem}
%[Follows from majorization theory of~\cite{HardyLP}; also Claim 2.6
%in~\cite{ibrahimpur2021minimum}]
\label{thm:Majorization-orig} 
\label{thm:controllAllL} \label{thm:majorization}
Let $x, y \in \mathbb{R}^n_+$, $\rho,\beta\geq 0$. 
Let $h:\mathbb{R}^n \mapsto \mathbb{R}_+$ be a monotone, symmetric norm.  
\begin{enumerate}[(a)] %[label=(\alph*), topsep=0.2ex, itemsep=0.1ex, leftmargin=*]
\item {\upshape\textsf{[Follows from majorization theory of~\cite{HardyLP}]}}
If $\topl(x) \leq \rho \cdot \topl(y) + \beta$ for all $\ell \in [n]$, then 
$h(x) \leq \rho \cdot h(y) + \beta \cdot h(1, 0,\ldots, 0)$. 

\item {\upshape\textsf{[Slight generalization of Claim 2.6 in~\cite{IbrahimpurS21}]}}
If $\topl(x) \leq \rho \cdot  \topl(y)+\beta$ for all $\ell \in \pos_{n, \delta}$, then
$h(x) \leq (1+\delta)\bigl(\rho\cdot h(y)+\beta\cdot h(1,0,\ldots,0)\bigr)$. 
\end{enumerate}
\end{theorem}

\begin{proofof}{part (b)}
Claim 2.6 in~\cite{IbrahimpurS21} states that if $u,v\in\R^n$ are such that
$\topl(u)\leq\topl(v)$ for all $\ell\in\pos_{n,\dt}$, then $h(u)\leq(1+\dt)h(v)$.
Now take $u=x$, and $v=\rho y^\downarrow+\beta\cdot (1,0,\ldots,0)$.
Then $\topl(v)=\rho\topl(y)+\beta$ for all $\ell\in[n]$, and 
$h(v)\leq\rho h(y)+\beta\cdot h(1,0,\ldots,0)$.
%Consider any $i\in[n]$, and let $\ell$ be the largest index in $\pos$ that is at most
%$i$. Then, $i\leq (1+\dt)\ell$. So $\topl[i](x)\leq(1+\dt)\topl(x)\leq(1+\dt)\
\end{proofof}
  
%\begin{theorem}[Claim 2.6 in \cite{ibrahimpur2021minimum}]
%\end{theorem}
  
%In our algorithm, 
We will need to estimate $\topl(u)$ for $\ell\in\pos_{n,\dt}$, given estimates
of $u^\downarrow_\ell$ for all $\ell\in\pos_{n,\dt}$. We utilize the following result
from~\cite{IbrahimpurS21}, which we have slightly paraphrased.
We say that a vector $v\in\R_+^\pos$ is non-increasing if
$v_\ell\geq v_{\ell'}$ for indices $\ell,\ell'\in\pos, \ell<\ell'$.

\begin{lemma}[Lemma 2.8 (a) and (b) in~\cite{IbrahimpurS21} slightly paraphrased]
\label{toplbnds}
Let $u\in\R_+^n$ and $v\in\R_+^{\pos_{n,\dt}}$ be a non-increasing vector. 
%i.e., $v_\ell\geq v_{\ell'}$ for indices $\ell,\ell'\in\pos, \ell<\ell'$.
For $r\in[n]$, let $\prev(r)$ denote the largest index $\ell\leq r$ in $\pos_{n,\dt}$.
Suppose that $u^\downarrow_\ell\leq v_\ell\leq (1+\ve)u^\downarrow_\ell+\kp$ for all
$\ell\in\pos_{n,\dt}$.  
%Define the vector $v^\exp:=(v_{\prev(r)})_{r\in[n]}$. 
Then $\topl(u)\leq\sum_{r=1}^{\ell}v_{\prev(r)}\leq(1+\ve)(1+\dt)\topl(u)+\ell\kp$ for all
$\ell\in[n]$. 
\end{lemma}

Complementing the above, we utilize the following result, which shows how we can find a
suitable set $\T$ containing good estimates of $u^\downarrow_\ell$ for all $\ell\in\pos_{n,\dt}$.

\begin{lemma}[Lemma 2.9 (b) in~\cite{IbrahimpurS21} paraphrased]
\label{polyenum}
Let $u\in\R_+^n$ and $0<\e\leq 1$. 
Suppose that we have bounds $\lb$ and $\ub$
such that $\lb\leq u^\downarrow_\ell\leq\ub$ for all $\ell\in\pos_{n,\dt}$. Let
$N=(2e)^{|\pos_{n,\dt}|}+\bigl(\frac{\ub}{\lb}\bigr)^{O(\frac{1}{\e})}$. We can construct
$\T\sse\R_+^{\pos_{n,\dt}}$ with $|\T|\leq N$ in $O(N)$ time that contains a non-increasing vector 
$v\in\R_+^{\pos_{n,\dt}}$ such that $u^\downarrow_\ell\leq v_\ell\leq(1+\e)u^\downarrow_\ell$ for
all $\ell\in\pos_{n,\dt}$. 
\end{lemma}

\subparagraph*{Minimum-norm $k$-clustering.}
In $k$-clustering problems, we are given a set of centers $\F$,  $n$ clients, $\C$, a
distance function $d : (\F \cup \C)^2 \to \mathbb{R}_{\ge 0}$, and must choose a set of
$k$ centers from $\F$ to open. The assignment cost incurred by client $j \in \C$ for a
given set  $S \subseteq \F$, denoted $d(j, S)$, is their distance to the closest center in
$S$. Thus, a set of open centers $S$ induces a cost vector, denoted $d(\C, S)$, where
$d(\C, S)_j = d(j, S)$ is the assignment cost incurred by client $j$ under $S$. 
We mostly consider the setting where $\F = \C$, but in Section~\ref{fneqc}, we discuss how
the general setting can be handled.%unless specified otherwise.  
%We primarily study two clustering problems in this paper. In the \emph{$\ell$-centrum}
%problem, the objective is to minimize the $\topl$-norm of $d(\C, S)$, the assignment cost
%vector induced by opening centers in $S$. This problem captures both the $k$-center ($\ell
%= 1$) and $k$-median ($\ell = n$) problems, and is a special case of the
In the \emph{minimum-norm $k$ clustering} problem, we seek to minimize $f\bigl(d(\C, S)\bigr)$,
where $f$ is a given {\em monotone, symmetric norm}. A useful special case of this is when
$f$ is a $\topl$-norm. Observe that this special case, which is sometimes called the 
{\em $\ell$-centrum problem}, captures both the $k$-center ($\ell= 1$) and $k$-median
($\ell = n$) problems.  

Recall that, for $k$-median, the following adaptive-sampling process yields a
constant-factor bicriteria approximation: pick a random point $\cli\in\C$ to add to the
current center-set $S$ with probability proportional to $d(\cli,S)$, and repeat this for
$O(k)$ iterations. 
%If we do this for $k$ iterations, then~\cite{ArthurV} showed that we
%obtain an $O(\log k)$-approximation, while  
%We preface our adaptive-sampling approach, as in~\cite{aaai25,aaai25b}
%A natural question is whether this algorithm, as is, also yields good solutions for
%the $\ell$-centrum problem as well. 
It is worth pointing out 
%As illustrated below, %by the following theorem, 
%as stated, 
%We show in Appendix~\ref{append-prelim}
that this algorithm, as stated, can fail badly even for $k$-center (i.e., $1$-centrum).
%problem. 
%there is a
%family of instances of the $k$-center problem (minimizing the $\mathrm{Top}_1$-norm) for
%which this adaptive sampling algorithm can yield very poor solutions with arbitrarily high
%probability.  

\begin{theorem}\label{thm:dSample-badEx}
Let $\tau \geq 1$, $L > 1$, $\epsilon \in (0, 1)$ be fixed. There exists a $2$-center
instance such that $\Pr[\topl[1]\bigl(d(\mathcal{C}, S)\bigr) > L \cdot \opt] \ge 1 -
\epsilon$, where $\opt$ is the optimal value, and $S$ is the set of $2\tau$ centers
obtained by running the $k$-median adaptive-sampling algorithm $2\tau $ times.  
\end{theorem}

\begin{proof}
Consider an instance with client-set  $\C = \C_1 \cup \{j^*\}$, where the set $\C_1$
contains at least $2\tau\left(1   + \frac{4\tau L}{\epsilon}\right)$ clients. Furthermore,
the distance between any pair of clients in $\C_1$ is 1. The client $j^*$ is located at a
distance of $L$ from every client in $\C_1$; notice that this indeed defines a valid
metric. An optimal solution for the $2$-center problem would be to open one center in
$\C_1$, and one center at $j^*$; this solution has a cost of 1. In fact, any set of
centers with cost less than $L \cdot \opt$ must include $j^*$ and some client in $\C_1$.
%For any $S \subseteq \mathcal{C}$, if $\mathrm{Top}_1(d(\mathcal{C}, S)) < L = L \cdot
%\opt$, then $d(j^*, S) < L$; but since $d(i, j^*) = L$ for all $i \in \mathcal{C}
%\setminus \{j^*\}$, this is only possible if $j^* \in S$.  
	
We now argue that, with probability at least $1 - \epsilon$, the $k$-median adaptive
sampling algorithm fails to open a center at $j^*$. Suppose no center has been opened at
$j^*$ by the end of the $(i - 1)$-th iteration. Let $S_{i-1}$ be the set of currently open
centers, and let $s_i$ be the center opened in iteration $i$. Since $j^*$ is at a distance
of $L$ from $S_{i-1}$,  $\Pr[s_i = j^* | j^* \notin S_{i-1}] = \frac{L}{|C_1| - |S_{i-1}|
  + L} \leq \frac{L}{|C_1| - 2\tau + L} < \frac{\epsilon}{8\tau^2}$; by applying the Union
bound, we have $\Pr[ j^* \in  S | j^* \notin S_{i-1}] \le \epsilon/4\tau  $. So, the
probability that $j^* \in S_{2\tau}$, is at most $\Pr[s_0 = j^*] + 2\tau \cdot
\frac{\epsilon}{4\tau} = \frac{1}{n} + \frac{\epsilon}{2} < \epsilon$. %Pr[s_0 = j^*] +
%\sum_i Pr[s_i = j^* | j* \notin S_{i-1}]  \cdot Pr[j^*\not in S_{i-1}]  
Since  any set of centers with cost less than $L \cdot \opt$ must include $j^*$, 
$\Pr[\topl[1]\bigl(d(\mathcal{C}, S)\bigr) > L \cdot \opt] \ge 1 - \epsilon$.  
%\hfill \qed
\end{proof}

\section{Adaptive sampling}\label{sec:msn} \label{adsample}
We now present our adaptive-sampling algorithm for minimum-norm $k$-clustering. As
suggested by Theorem~\ref{thm:majorization}, our approach for controlling a general monotone,
symmetric norm is by controlling all $\topl$ norms. So %$\topl$-norms play a key role, and
we first discuss adaptive sampling for $\topl$-norms in Section~\ref{adsample-topl}. We
then exploit insights gained from the {\em analysis} of adaptive sampling for $\topl$
norms to devise our adaptive-sampling algorithm for monotone, symmetric norms
(Section~\ref{adsample-gen}). 

\subsection{\boldmath $\topl$-norms} \label{adsample-topl}
We give an overview of the key ideas of the adaptive-sampling algorithm of \cite{aaai25}
for $\topl$-norms, which we build upon in Section~\ref{adsample-gen}. 
As noted previously, the $\topl$-objective can be difficult to work
with directly, due to its non-separable nature. Instead, \cite{aaai25} observed
that, for suitably chosen parameters $\beta, t_\ell$, one can work with the proxy function
$\sum_{j\in\C}(d(\cli,S)-\beta  t_\ell)^+$ suggested by Claim~\ref{csproxy} 
%of Chakrabarty and Swamy \cite{ChakrabartyS19}:  
(where $S$ is the center-set), which
{\em is} separable. 
We can now treat $(d(\cli, S) - \beta t_\ell)^+$ as the individual contribution of client
$\cli$ to the current cost of the solution, and so the next center
to open is sampled with probability proportional to this {\em proxy cost}.  

So the adaptive-sampling algorithm in~\cite{aaai25} for $\topl$-norms is simply the
following: choose the first center uniformly at random from $\C$, choose every
subsequent center as discussed above, and continue this for $O(k)$ centers. They showed
that if $t_\ell$ is good estimate of the $\ell$-th largest cost incurred by an optimal
solution, then (for a suitably chosen $\beta$) this returns an $O(1)$-approximation with
constant probability. 

\subparagraph*{Analysis.}
We include some components of the analysis in~\cite{aaai25b}, which we need to modify
slightly for our purposes. Throughout, $S$ denotes the current (random) center-set.
Fix an optimal solution $\optSoln=\{\optcen_1,\ldots,\optcen_k\}$ for the $\topl$
$k$-clustering problem, and let $\clustroot_1, \ldots,\clustroot_k$ denote the clusters
induced by this solution, which we will often call the optimal clusters;  
that is, $\clustroot_q$ is the set of clients assigned to center $\optcen_q$, for all
$q\in[k]$. Let $\opt_\ell=\topl\bigl(d(\C,\optsoln)\bigr)$ denote the optimal
$\topl$-cost. Let $\tau=\tauval$, $\rho=\rhoval$. 
It will be convenient to state the analysis in terms of certain parameters
$\al,\beta,\gm$, whose values are chosen to satisfy the following:%
\footnote{These imply the inequalities (3) in~\cite{aaai25b}, which are
used in their analysis of adaptive sampling for $\topl$-norms, by taking $\kp=\al+\beta+3$.}
\begin{equation}
%  \begin{split}
        \beta = \gm = \alpha +1, \quad
        \rho\geq 2(\al+2\beta+3)+\gm, \quad %\kp\geq 2\al+\beta+2, \\
        %1-\tfrac{\gamma }{\rho}\geq 2\cdot\tfrac{\kappa + \beta }{\rho} \\
        %\kappa  & \geq 2\alpha + \beta  +2, \qquad
        \Bigl(1-\tfrac{\gamma}{\rho}\Bigr)\cdot\tfrac{\alpha-1}{2\alpha(\al+\beta+3)}\geq\frac{1}{\tau}.
%    \end{split}
\label{abkchoices}
\end{equation}
For instance, one can take $\al=2$, $\gm=\beta=\betaval$. 
%$\alpha=3$, $\kappa=11$, $\tau=38$, and $\rho=35$.
At a given iteration of the algorithm, we classify an optimal cluster 
$\clust \in\bigl\{\clustroot_1,\ldots,\clustroot_k\bigr\}$ as ``good'' or ``bad'',
depending on whether the clients in $\clust$ incur a large proxy cost compared to what
they incur under $\optSoln$.   
%A cluster $\clust$ is ``bad", if 

\begin{definition} \label{clgood}
Say that $\clustroot_q$ is $\ell$-\emph{good}, if 
$\sum_{j \in \clustroot_q}(d(\cli, S)-\beta  t_\ell)^+\leq
\gamma \sum_{j \in\clustroot_q}(d(\cli, \optcen_q) - t_\ell)^+$. 
If $\clustroot_q$ is not $\ell$-good, we say that it is {\em $\ell$-bad}.
\end{definition}

If all clusters are good, then the $\topl$-cost of the current solution is $O(\opt_\ell)$,
assuming that $t_\ell$ is a good estimate of $t^*_\ell = d(\C, \optSoln)^\downarrow_\ell$.
%(see Claim~\ref{allgood}).  

\begin{claim}[Claim 4.24 in~\cite{aaai25b}] \label{allgood}
Let $t^*_\ell=d(\C,\optSoln)^\downarrow_\ell$. 
%be the $\ell$-th largest assignment cost incurred under $\optSoln$. 
Suppose that $t^*_\ell\leq t_\ell\leq\max\bigl\{(1+\varepsilon)t^*_\ell,\frac{\ve\opt_\ell}{\ell}\bigr\}$.
If every cluster is $\ell$-good, then 
$\topl(d(\mathcal{C},S))\leq(1+\varepsilon)\gamma\cdot\opt_\ell$.
\end{claim}

\begin{proof} %{Claim~\ref{allgood}}
We include the proof from~\cite{aaai25b} to keep the exposition more self-contained.
We have 
\begin{equation*}
\begin{split}
\topl\bigl(d(\C,S)\bigr) & \leq\ell\cdot\beta t_\ell+
\sum_{q=1}^k\sum_{\cli\in \clust}(d(\cli,S)-\beta t_\ell)^+ \\
& \leq \gm\cdot\max\bigl\{(1+\ve)\ell t^*_\ell,\,\ve\cdot\OPT_\ell\bigr\}+
\sum_{q=1}^k\gm\cdot\sum_{\cli\in\clust}(d(\cli,\optcen_q)-t^*_\ell)^+ \\
& = \gm\cdot\max\bigl\{(1+\ve)\ell t^*_\ell,\,\ve\cdot\OPT_\ell\bigr\}+
\gm\cdot\sum_{\cli\in\C}(d(\cli,\optsoln)-t^*_\ell)^+\leq \gm(1+\ve)\OPT_\ell.
\end{split}
\end{equation*}
The first inequality is from Claim~\ref{csproxy} (b). 
The second inequality is because all clusters are $\ell$-good, we have $\gm\geq\beta$, and
due to the bounds on $t_\ell$.
\end{proof}

%We include the proof in Appendix~\ref{append-adsample} to keep exposition more
%self-contained. 
On the other hand, if the $\topl$-cost of our solution is large, then one can show that
the next center added to the solution $S$ lies in the 
``core'' of some bad cluster $\clust$, with some constant probability $p$. Moreover, when 
a center is opened from the core of $\clust$, that cluster becomes good
(Claim~\ref{coreisgood}) and hence, remains good in all subsequent iterations. 
Definition~\ref{clcore} gives the precise definition of core, %which is tailored to ensurethe latter property, 
but informally, the core of a cluster
$C^*_q$ consists of points that are sufficiently close to its center $\optcen_q$.
%The formal definition of core 

Thus, in every iteration, we make progress towards obtaining a low-cost
solution by reducing the number of bad clusters with probability $p$. Hence, the expected number
of bad clusters decreases with each iteration, and one can then use standard
martingale arguments to show that after $(k+\sqrt{k})/p$ iterations, with some constant
probability, we obtain a solution with no bad clusters, and therefore a near-optimal
solution. 
We now discuss some details.

Define the {\em radius} of a cluster $\clustroot_q$ to be 
$r_\ell(\clustroot_q) := \frac{\sum_{j \in \clustroot_q} (d(\cli, \optcen_q)-t_\ell)^+}{|\clustroot_q|}$. 
 
\begin{definition} \label{clcore}
The $\ell$-core of %a cluster 
$\clust$ is %defined as  
$\lcore(\clust):=\{j \in \clust : d(\cli, \optcen_q) \leq \alpha(t_\ell + \rl(\clust)) \}$.  
\end{definition}

\begin{remark} \label{diffcore}
Our definition of core is slightly different from (and cleaner than) 
the definition in~\cite{aaai25b} (see Definition 4.26 in~\cite{aaai25b}), where this is
defined as  
$\{\cli\in\clust: d(\cli,\optcen_q)\leq t_\ell\}$ if $\optcen_q$ is sufficiently close to $S$,
and as $\bigl\{\cli\in\clust: (d(\cli,\optcen_q)-t_\ell)^+\leq\al\rl(\clust)\bigr\}$ otherwise. 
The reason for this change will become clear in Section~\ref{adsample-gen}.
The precise definition used in~\cite{aaai25b} is not so important, but 
{\em importantly}, we note that the core of $\clust$, as defined in~\cite{aaai25b}, is a 
{\em subset} of $\lcore(\clust)$ as defined above. %in Definition~\ref{clcore}.
For clarity, we use {\em $\ell$-supercore} to refer to the version of core defined
in~\cite{aaai25b}; so we have $\supcore_\ell(\clust)\sse\lcore(\clust)$.
\end{remark}

By slightly reworking the arguments in~\cite{aaai25} (to account for our definition of
core), we have the following. 

\begin{claim} \label{coreisgood}
Consider a cluster $\clust$. 
%and let $S$ be the current center-set.  
%Let the parameters $\al,\beta,\gm$ in Definitions~\ref{clgood} and~\ref{clcore} satisfy
%$\gm\geq\beta\geq\al+1$. 
If $S\cap\lcore(\clust)\neq\varnothing$, then $\clust$ is $\ell$-good (and hence remains
$\ell$-good throughout).  
\end{claim}
   
\begin{proof} %{Claim \ref{coreisgood}}
%we show below that if we have a 
%center from $\lcore(\clust)$, then $\clust$ is $\ell$-good. 
Let $s$ be a point in $S\cap\lcore(\clust)$. Since {$\beta\geq\alpha  + 1$}, we have
\begin{alignat*}{1}
\sum_{j \in \clust}(d(\cli, s) - \beta t_\ell)^+  &\leq \sum_{j\in \clust}\bigl((d(\cli, \optcen_q)-t_\ell)^++(d(s,\optcen_q)- \alpha t_\ell)^+\bigr) \\ 
&\leq |\clust| \cdot \bigl(\rl(\clust) +(\alpha t_\ell + \alpha \rl(\clust) - \alpha t_\ell)^+\bigr)
= (\alpha + 1) \cdot |\clust| \cdot \rl(\clust).
\end{alignat*}
The first inequality is because $d$ satisfies the triangle inequality, and since
$(y+z)^+\leq y^++z^+$.
%The second inequality is because $t_\ell\geq t^*_\ell$. 
The second inequality follows from the definition of $r_\ell(\clust)$ and since  $s\in\lcore(\clust)$. 
%the second term in the final inequality above is at most 
%$\alpha |C^*_q|r_\ell(C^*_q)$; note that this holds both when $C^*_q$ is $\ell$-close and
%is $\ell$-far. 
%So we have 
%$\sum_{j \in C^*_q}(d(\cli,s) - \beta t_\ell)^+\leq 
%(1 + \alpha)\sum_{j \in C^*_q} (d(\cli, \optcen_q) - t^*_\ell)^+$, 
Finally, since $\gamma\geq\alpha+1$ (see \eqref{abkchoices}), this shows that $C^*_q$ is
$\ell$-good. 
\end{proof}

Lemma~\ref{choosecore} %stating that we hit the core of some bad cluster in any iteration
%where the current cost us large, 
is the main result proved in~\cite{aaai25b}, which underlies their proof that adaptive
sampling works for $\topl$-norms.
This will also be key to extending adaptive sampling to monotone,
symmetric norms in Section~\ref{adsample-gen}.

\begin{lemma}[Follows from Lemma 4.28 in~\cite{aaai25b}]
\label{lem:chooseFromCore} \label{choosecore}
Consider any iteration $i$ for which
$\topl(d(\C, S_{i-1}))>\rho {(1 + \varepsilon)} \cdot \topl(d(\C,S^*))$. 
%If the next center then 
%Under the above sampling process, where 
Suppose we open the next center $s_i$ at $\cli\in\C$ with probability
proportional to $(d(\cli,S_{i-1}) - \beta t_\ell)^+$. Then, $s_i$ lies in the
($\ell$-supercore, and hence) $\ell$-core of some $\ell$-bad cluster $\clust$, 
%and hence in $\lcore(\clust)$, 
with probability at least $\frac{1}{\tau}$.
%$\Pr[\text{$s_i$ lies in the core of an $\ell$-bad cluster}]\geq\frac{1}{\tau}$.
%, with constant probability.   
\end{lemma}
 
% CONSTANTS FOR THIS ANALYSIS: beta = 3 = gamma, alpha = 2, kappa = 11, rho = 35, tau = 50

\subsection{Arbitrary monotone symmetric norms} \label{adsample-gen}
We now describe the adaptive-sampling algorithm for 
\emph{minimum-norm $k$-clustering}. Recall that in this problem, we seek to minimize 
$f\bigl(d(\C, S)\bigr)$, where $d(\C, S)$ is the cost vector induced by opening $S$, and
$f:\R^n\mapsto\R_+$ is an arbitrary monotone, symmetric norm. We will typically refer to 
$f(d(\C, S))$ as the $f$-cost induced by $S$. Let $\pos=\pos_{n,\dt}$ 
(Definition~\ref{posdef}). 

Again, the objective function %is difficult to work with due to its non-separable nature;
is non-separable in that we cannot isolate the contribution of a specific
client to the $f$-cost in any convenient way. 
%the objective function aggregates clients' costs in a way that we 
%contribution of a client to the cost depends also on the other clients' costs. 
Furthermore, %as noted earlier, 
{in contrast with $\topl$-norms, %there is no convenient
we do not have a 
proxy-cost function %that we can work with, 
that renders the problem separable}.
%
%In the $\topl$ setting, we were able to circumvent this issue by working with a separable
%proxy function for the $\topl$ norm, and using it to bias our sampling protocol
%suitably. 
%While we no longer have such a separable proxy for the arbitrary $f$-norm in
%question, we instead utilize 
We therefore resort to Theorem~\ref{thm:controllAllL}, which  
%The upshot of Theorem \ref{thm:majorization} is that, 
shows that in order to control the $f$-cost induced by our solution, it suffices to
control the $\topl$-cost of our solution with respect to all $\ell \in \pos$.  
%(Recall the definition of $\pos$ in Definition~\ref{posdef}.)
%that 
%$\pos=\pos_{n, \delta}:=\{\min\{\lceil(1 + \delta)^s\rceil, n\} : s \geq 0\}$.) 

A naive way of utilizing this observation is to 
%This observation is particularly helpful, as we already have an adaptive-sampling algorithm for
%the $\topl$-objective \cite{aaai25}. It is tempting to consider the following naive
%algorithm: 
simply run the adaptive-sampling algorithm for $\topl$-norms from
Section~\ref{adsample-topl} for each $\ell \in \pos$, and return the union of these
solutions. %There are two issues with this naive approach. 
However, this would open $O(k|\pos|)=O(k\log n)$ centers, instead of the $O(k)$  
centers we seek in a $\bigl(O(1),O(1)\bigr)$-approximate solution. 
%Examining this naive algorithm more closely, one realizes that in fact, 
Also, this naive algorithm would achieve much more than what is required in that
%is overambitious in the guarantee it provides because
%this would yield a {\em much} stronger property than we need: 
%we would have 
the resulting solution $S$ would satisfy
$\topl\bigl(d(\C, S)\bigr) \leq \rho (1 + \varepsilon) \opt_\ell$ for every 
$\ell \in \pos$,%
\footnote{The probability of success in one run would only be $e^{-O(\log n)}$, 
%for some constant $p$, 
but this can be boosted to constant probability of success by repeating the
process $n^{O(1)}$ times.}
where recall that $\opt_\ell$ is the optimal $\topl$-cost, 
%But this would imply that for 
which implies that it is a $\rho(1+\ve)$-approximate solution for 
{\em every monotone, symmetric norm $h$} (by Theorem~\ref{thm:majorization}); 
%$h\bigl(d(\C, S)\bigr) \leq \rho (1 + \varepsilon)\cdot(\text{optimum for norm $h$})
this kind of universal approximation is a {\em much stronger} guarantee than we seek, and 
for such a strong guarantee one {\em must} end up opening $O(k\log n)$
centers~\cite{kumar2000fairness}. 

Since we only want the $f$-cost $f\bigl(d(\C,S)\bigr)$ of our solution to be
$O(1)\cdot\OPT$, where $\OPT$ is the optimal $f$-cost, %and for this, 
we only need that, for every $\ell\in\pos$, $\topl\bigl(d(\C,S)\bigr)$ is within
a constant-factor of the $\topl$-cost of a {\em fixed} solution, namely the optimal
solution $\optSoln$ to the minimum-norm problem. 
We design a sampling process that allows us to do so by
capitalizing on some insights from the analysis of the adaptive-sampling algorithm
for $\topl$-norms presented in Section~\ref{adsample-topl}. From the analysis therein, we
can infer that to ensure that 
$\topl\bigl(d(\C,S)\bigr)=O(1)\cdot\topl\bigl(d(\C,\optSoln)\bigr)$ for all $\ell\in\pos$,  
it suffices for $S$ to hit the $\ell$-core of every cluster $\clust$ induced by
$\optSoln$, for all $\ell\in\pos$. Here, one extremely useful property comes in handy:
for any fixed cluster $\clust$, 
{\em the $\ell$-cores for different indices $\ell$ are nested}
(Claim~\ref{corenested}); this is precisely, why we define $\lcore(\clust)$ slightly
differently as compared to~\cite{aaai25b}. 
This implies that if we open a center from $\lcore(\clust)$, then we have also hit
$\lcore[\ell'](\clust)$ for all indices $\ell'\leq\ell$; so we have {\em simultaneously}
turned $\clust$ into an $\ell'$-good cluster (see Definition~\ref{clgood}) for all %indices
$\ell'\leq\ell$, and thereby made progress by bounding the proxy-cost of $\clust$ for all
these indices.
%i.e., we have 
%$\sum_{j\in\clust}(d(j,S)-\beta t_{\ell})^+\leq\gm\sum_{j\in\clust}(d(j,\optSoln)-t_{\ell'})^+$ 
%of $\clust$ with respect to $S$ terms of 
%for all $\ell'\leq\ell$.

We can now finish things up with one additional observation: we only need to consider
indices $\ell\in\pos$ for which the current $\topl$-cost %$\topl\bigl(d(\C,S)\bigr)$ 
is large. Formally, call an index $\ell$ {\em \costly} if
$\topl\bigl(d(\C,S)\bigr)>\rho(1+\ve)\topl\bigl(d(\C,\optSoln)\bigr)$. (In the actual
algorithm, we work with a suitable estimate of $\topl\bigl(d(\C,\optsoln)\bigr)$.)
Let $\ell_{\max}$ be the largest \costly index. Now suppose we run an iteration of the
adaptive-sampling algorithm from Section~\ref{adsample-topl} for index $\ell_{\max}$, i.e., we 
open the next center at $\cli\in\C$ with probability proportional to 
$(d(j,S)-\betaval t_{\ell_{\max}})^+$. By Lemma~\ref{choosecore}, we know that we will hit
$\lcore[{\ell_{\max}}](\clust)$ for some $\ell_{\max}$-bad cluster $\clust$ with constant
probability, which will imply that {\em $\clust$  becomes an $\ell$-good cluster 
for every \costly index $\ell\in\pos$}. The upshot is that we have now
``taken care of $\clust$'' in the sense that for every index $\ell\in\pos$, we know that
the $\topl$-proxy cost of clients in $\clust$ is small, or the $\topl$-cost of the entire
solution is small. Therefore, by Lemma~\ref{choosecore}, we obtain that with some
constant probability $p$, we ``take care of'' some cluster in each iteration, and one can
utilize the same martingale argument to show that with constant probability, after $O(k)$
iterations, there is no \costly index; hence, $f\bigl(d(\C,S)\bigr)=O(1)\cdot\OPT$
(assuming we have the right $t_\ell$ values).
%
%We include below the precise description of the algorithm.

\begin{algorithm}[h!]
    \caption{\qquad Adaptive sampling for minimum-norm $k$-clustering} 
                \label{alg:adsample-min-norm} \label{adsample-alg}
    \textbf{Input}: instance $(\mathcal{C}, d)$, %parameters $\tau\geq 1$, and 
    a non-increasing vector $\vec{t} \in \mathbb{R}^{|\pos|}_+$ 
    %(i.e., $t_\ell\geq t_{\ell'}$ if $\ell,\ell'\in\pos, \ell<\ell'$)
    %$\vec{t}=\vec{t}^\downarrow$). 
    %(For $r\in[n]$, let $\prev(r)$ be the largest index $\ell\leq r$ in $\pos$.) 
    \begin{algorithmic}[1]
        \STATE Initialize $S_0 \leftarrow \varnothing$, $\indset\assign\pos$
        \STATE For $\ell\in\pos$, define $B_\ell:=\sum_{r=1}^\ell t_{\prev(r)}$, where 
        $\prev(r)$ is the largest index $\ell\leq r$ in $\pos$ \label{blestim}
        %\STATE $\ell_{\max} \leftarrow n$ 
        \FOR{$i=1,\ldots,\ceil{\tauval(k + \sqrt{k})}$}
        \STATE If $\indset=\varnothing$, set $\ell_{\max}\assign 1$; otherwise 
        set $\ell_{\max}\assign$ largest index in $\indset$. 
        \STATE If $i=1$, sample $s_i$ uniformly at random from $\C$; otherwise, sample
        $s_i$ with probability proportional to 
        $\left(d(s_i, S_{i-1}) - \betaval t_{\ell_{\max}}\right)^+$. \label{sample} \\  
        \STATE Update $S_i \leftarrow S_{i-1} \cup \{s_i\}$. \label{fopen}
        \STATE Update $\indset\assign\bigl\{\ell\in\pos: 
        \topl\bigl(d(\C,S)\bigr)>\rho(1+\ve)\cdot B_\ell\bigr\}$. 
        %\STATE Update $\ell_{\max} \leftarrow \max\{ \ell \in \pos : \topl(d(\C, S)) > \rho (1 + \varepsilon) \cdot \topl(d(\C, \optSoln))\}$
        \ENDFOR 
        \RETURN $S_{\ceil{35(k + \sqrt{k})}}$
    \end{algorithmic}
\end{algorithm}

\subparagraph*{Analysis.}
It will be convenient to assume via scaling that $f(1,0,\ldots,0)=1$. 
For $\ell\in[n]$, let $t^*_\ell=d(\C,\optSoln)^\downarrow_\ell$ be the $\ell$-th largest
assignment cost incurred under the optimal solution $\optSoln$. 
Our main result is as follows. 
%We prove the following result.
We have not attempted to optimize constants here, but instead chosen simplicity of
exposition (and some calculations).

\begin{theorem}\label{thm:min-norm} \label{adsample-thm}
Let $\vec{t}$ be a non-increasing vector %= (t_1, \ldots, t_n) \in \R^{|\pos|}$ 
such that $t^*_\ell \leq t_\ell \leq\max\{(1 + \veps) t^*_\ell, \frac{\veps\opt}{n}\}$ 
for all $\ell \in \pos$. 
%where $\ve\leq 1$.
%Let $S$ be the set of centers opened by Algorithm \ref{alg:adsample-min-norm} and 
Algorithm~\ref{alg:adsample-min-norm} runs in $O(nk)$ time and returns a set $S$ of $O(k)$
centers satisfying 
$f\bigl(d(\C, S)\bigr)\le\rhoval(1 + \varepsilon)^3(1 + \delta)^2\cdot \opt$ 
with constant probability.   
\end{theorem}

Complementing this, Lemma~\ref{polyset} shows that one can compute in polytime a set $\T$
containing a vector $\vec{t}$ satisfying the conditions of Theorem~\ref{adsample-thm}. So
if we run Algorithm~\ref{adsample-alg} for each vector in $\T$ and return the best solution
obtained, then this will be an $\bigl(O(1),O(1)\bigr)$-bicriteria solution for 
minimum-norm $k$-clustering with constant probability.

In Algorithm~\ref{adsample-alg}, we use $B_\ell:=\sum_{r=1}^\ell t_{\prev(r)}$ as an estimate of
$\topl\bigl(d(\C,\optSoln)\bigr)$. We begin by showing that this estimate is
fairly accurate if $\vec{t}$ satisfies the conditions of Theorem~\ref{adsample-thm}. 

\begin{claim} \label{toplestim}
Suppose that $t^*_\ell \leq t_\ell \leq
\max\{(1 + \veps) t^*_\ell, \frac{\veps\opt}{n}\}$ for all $\ell \in \pos$.
Then, we have 
$\topl\bigl(d(\C,\optSoln)\bigr)\leq B_\ell\leq(1+\ve)(1+\dt)\topl\bigl(d(\C,\optSoln)\bigr)+\ve\opt$
for all $\ell\in\pos$.
\end{claim}

\begin{proof}
%This follows from the following result in~\cite{IbrahimpurS21}, which we paraphrase slightly. 
Apply Lemma~\ref{toplbnds} to $u=(t^*_\ell)_{\ell\in[n]}$ and $v=(t_\ell)_{\ell\in\pos}$.
%completes the proof.
\end{proof}

Next, we show that the $\ell$-cores of a cluster are nested.  
%This follows from the fact
%that $t_\ell+r_\ell(\clust)$ is a non-increasing function of $\ell$.  

\begin{claim} \label{corenested}
For any $\ell' \leq \ell$, we have $\lcore(\clust) \subseteq \lcore[\ell'](\clust)$. 
\end{claim}

\begin{proof} %{Claim \ref{corenested}}
  This follows from the fact that $t_\ell+r_\ell(\clust)$ is a non-increasing function of
  $\ell$. (Note that $t_{\ell}\leq t_{\ell'}$.) 
  Given this, the result is immediate since 
  $\cli \in \lcore(\clust)$ implies that 
  $d(\cli, \optcen_q) \leq \alpha(t_\ell+r_\ell(\clust))\leq\alpha(t_{\ell'}+r_{\ell'}(\clust))$, 
  and so $\cli\in\lcore[\ell'](\clust)$.
  To see the fact, we have 
\begin{equation*}
\begin{split}
t_\ell + r_\ell(\clust) & = t_\ell+\frac{\sum_{\cli \in \clust} (d(\cli, \optcen_q)-t_\ell)^+}{|\clust|} 
=t_\ell+\frac{\sum_{\cli \in \clust}\bigl((d(\cli,\optcen_q)-t_{\ell'})+(t_{\ell'}-t_\ell)\bigr)^+}{|\clust|}
\\
& \leq t_\ell+\frac{\sum_{\cli \in
    \clust}(d(\cli,\optcen_q)-t_{\ell'})^+}{|\clust|}+(t_{\ell'}-t_\ell)
= t_{\ell'}+r_{\ell'}(\clust).
%\frac{\sum_{\cli \in \clust}(d(\cli,\optcen_q)-t_{\ell'})^+}{|\clust|}.
\end{split}
\end{equation*}
The inequality follows because $(y+z)^+\leq y^++z^+$.  
\end{proof}

\begin{proofof}{Theorem~\ref{adsample-thm}}
As in Section~\ref{adsample-topl}, let $\tau=\tauval$, $\rho=\rhoval$, 
$\al=2$, $\gm=\beta=3$. 
%let parameters $\al,\beta,\gm$ be defined to 
These parameters satisfy \eqref{abkchoices}. %e.g., we can take
For any index $\ell$, Definition~\ref{clgood} specifies what it
means for a cluster $\clust$ to be $\ell$-good or $\ell$-bad, and $\lcore(\clust)$ is
defined in Definition~\ref{clcore}.

Note that $\indset$ can only shrink as the algorithm proceeds.
Observe that if $\indset\neq\es$ at the start of some iteration, then there must be
some $\ell_{\max}$-bad cluster $\clust$; otherwise, by Claim~\ref{allgood} (and
since $\gm\leq\rho$), we would have 
$\topl[{\ell_{\max}}]\bigl(d(\C,S)\bigr)\leq\rho(1+\ve)\topl[{\ell_{\max}}]\bigl(d(\C,\optSoln)\bigr)\leq
\rho(1+\ve)B_{\ell_{\max}}$, contradicting that $\ell_{\max}\in\indset$.
So if $\indset\neq\es$, then by Lemma~\ref{choosecore}, 
%if $S=S_{i-1}$ is such that $\indset\neq\es$, then 
with probability at least $\frac{1}{\tau}$, 
the next center added to $S$ lies in $\lcore[{\ell_{\max}}](\clust)$ for some
$\ell_{\max}$-bad cluster $\clust$. Since the $\ell$-cores are nested, by
Claim~\ref{coreisgood} and the definition of $\ell_{\max}$, this implies that $\clust$
becomes $\ell$-good for every $\ell\in\indset$ after this iteration.

Define
$\bad:=\{q\in[k]: \clust\text{ is $\ell$-bad for some $\ell\in\indset$}\}$. The above
discussion shows that if $\indset\neq\es$ at the start of an iteration, then after the
iteration, $|\bad|$ has decreased by at least $1$ with probability at least
$p:=\frac{1}{\tau}$. 
Given this, we can follow the same martingale-based-approach used in~\cite{aaai25b} to
analyze adaptive sampling for $\topl$-norms.
%the proof of Theorem 4.14 in \cite{aaai25b}. 
Let %$p = 1/\tau$ and 
$N = \ceil{\tau(k+\sqrt{k})}$. 
Ideally, we would like to define $X_i=|\bad|$ for the set $\bad$ at the end of iteration
$i$, %as the number of clusters which are 
%$\ell$-bad for some \emph{uncontrolled} $\ell \in \pos$, (at the end of iteration $i$),
%and then consider a shifted version of this to obtain a supermartingale. 
but $X_{i} - X_{i+1}$ could potentially be large. 
So we define $X_0 = k$.
For $i \ge 1$, we define $X_i = X_{i-1} - 1$ if 
$\indset=\es$, or if $s_i$ lies in the $\ell_{\max}$-core of some $\ell_{\max}$-bad cluster
$\clust$. Otherwise, we define $X_i = X_{i-1}$. 
So we have $\E{X_i|X_{i-1}} \le X_{i-1} - p$ by Lemma~\ref{choosecore}.

Note that, by construction, if $\indset\neq\es$ at the start of an iteration $i$,  
then $|\bad|$ at the end of iteration $i$ is at most $X_i$.
So if $X_N=0$, we must have $\indset=\es$ after iteration $N$: if $\indset\neq\es$
at the start of iteration $N$, then $|\bad|\leq X_N=0$ at the end of iteration $N$,
which implies that $\indset=\es$ after iteration $N$.
But $\indset=\es$ means that 
\[
\topl\bigl(d(\C,S)\bigr)\leq\rho(1+\ve)B_\ell
\leq\rho(1+\ve)\Bigl((1+\ve)(1+\dt)\topl\bigl(d(\C,\optSoln)\bigr)+\ve\opt\Bigr) 
\quad\ \text{for all $\ell\in\pos$} 
\]
where the last inequality is due to Claim~\ref{toplestim}. 
By Theorem~\ref{thm:majorization} (b) (recall that $f(1,0,\ldots,0)=1$), this implies that 
$f\bigl(d(\C,S)\bigr)\leq\rho(1+\ve)^2(1+\dt)^2\cdot\opt+\rho(1+\ve)(1+\dt)\ve\opt
\leq\rho(1+\ve)^3(1+\dt)^2\opt$.

%$f(d(\C, S_N)) \le \rho(1 + \varepsilon) \opt$, or every cluster
%is $\ell$-good for all \emph{uncontrolled} indices $\ell$; in this latter case, by Theorem
%\ref{thm:majorization}, $f(d(\C, S_N)) \le \rho(1 + \varepsilon) \opt$. 
Finally, we argue that $\Pr[X_N > 0] \le e^{-p/4}$. For $i = 0, 1, \ldots$,
define $Y_i = X_i + i \cdot p$. Then $|Y_{i+1} - Y_i| \le 1$ and 
$\E{Y_{i+1} | Y_i} \le \E{X_{i+1}|X_i}+(i+1)p\leq X_i+i\cdot p=Y_i$, 
so $Y_0, Y_1, \ldots$ form a supermartingale. If $X_N > 0$, we
have $Y_N > N p$, so by the Azuma-Hoeffding inequality,  
\begin{equation*}
%  \begin{split}
  \Pr[Y_N-Y_0>(Np-k)] \leq\exp\Bigl(-\tfrac{(Np-k)^2}{2N}\Bigr) 
  =\exp\Bigl(-\tfrac{kp}{2(k+\sqrt{k})}\Bigr)\leq e^{-\frac{p}{4}}. 
%  \end{split}
\end{equation*}

The running time follows because each iteration of the {\bf for}-loop takes $O(n)$ 
time. 
\end{proofof}

%The only thing left to 
We conclude by showing that we can efficiently obtain a vector
$\vec{t}$ satisfying the conditions of Theorem~\ref{thm:min-norm}. 

\begin{lemma}\label{cor:polytimeEnum-ourSetting} \label{polyset}
We can obtain a set $\mathcal{T} \subseteq \mathbb{R}^{\pos}_+$ of size
$O(\frac{1}{\varepsilon}\cdot \log(n) \cdot
\max\{(\frac{n}{\varepsilon})^{O(1/\varepsilon)}, n^{1/\delta}\})$ that contains a valid
threshold vector $\vec{t}$ satisfying the conditions of Theorem \ref{thm:min-norm}.  
\end{lemma}

\begin{proof} %{Lemma \ref{cor:polytimeEnum-ourSetting}}
%We utilize the following result from~\cite{IbrahimpurS21}.
%which is Lemma 2.9 (b) from~\cite{IbrahimpurS21} paraphrased to suit our purposes.
We utilize Lemma~\ref{polyenum}.
Take $u$ to be the vector
$\bigl(\max\bigl\{t^*_\ell,\frac{\ve\opt}{2n}\bigr\}\bigl)_{\ell\in[n]}$, and set 
$\e=\min\{1,\ve\}$. We can now apply Lemma~\ref{polyenum} 
%and let $\vec{t}$ be the vector $v$ returned by the lemma, 
provided that we can obtain bounds $\lb$ and $\ub$ as required
by the lemma.
%In fact, we can show that in our setting, there exists a set $A$ of size
%$O(\frac{1}{\varepsilon} \log n)$ satisfying the conditions of Lemma
%\ref{lem:polytimeEnum}.
We run the well-known $2$-approximation algorithm for $k$-center due to
Gonzalez~\cite{Gonzalez} to obtain a center-set $S^*_1$. 
Note that $\opt\geq f(1,0,\ldots,0)\cdot\topl[1]\bigl(d(\C,\optSoln)\bigr)$ due to monotonicity
of $f$. Recall that we have normalized $f$ so that $f(1,0,\ldots,0)=1$. 
So we have $\opt\geq t^*_1$. We also have $t^*_1\geq\topl[1]\bigl(d(\C,S^*_1)\bigr)/2$ since
$S^*_1$ is a $2$-approximate $k$-center solution. So letting
$\lb:=\frac{\ve\cdot\topl[1](d(\C,S^*_1))}{4n}$, we obtain that 
$\max\bigl\{t^*_\ell,\frac{\ve\opt}{2n}\bigr\}\geq\lb$ for all $\ell\in\pos$. 

We also have $t^*_1\leq\opt\leq f\bigl(d(\C,S^*_1)\bigr)
\leq f(1,0,\ldots,0)\cdot\topl[n]\bigl(d(\C,S^*_1)\bigr)$, where the last inequality
follows from the triangle inequality.
So taking $\ub:=\max\bigl\{1,\frac{\ve}{2n}\bigr\}\cdot n\cdot\topl[1]\bigl(d(\C,S^*_1)\bigr)$, 
we obtain that $\max\bigl\{t^*_\ell,\frac{\ve\opt}{2n}\bigr\}\leq\ub$ for all
$\ell\in\pos$.

So by Lemma~\ref{polyenum}, we can compute in polytime $\T$ with
$|\T|\leq\bigl(\frac{n}{\ve}\bigr)^{O(\frac{1}{\ve})}$ 
containing the desired vector $\vec{t}$.
\end{proof}

\section{Refinements and extensions} \label{extn}

\subsection{An \boldmath $O(\log k)$-approximation for $\topl$ norms using $k$ centers}
\label{topl-lnk}
%In this section, 
We show that, for $\topl$ norms, if we stop the adaptive-sampling process
after $k$ iterations, then we obtain an (true) $O(\log k)$-approximation. Recall that such a
guarantee was obtained by Arthur and Vassilvistskii~\cite{ArthurV07} for $k$-means and
$k$-median, where the latter is the special case of $\topl$ norm when $\ell=n$. We
consider the setting $\F=\C$, but one can show that this guarantee also holds for the
extension of adaptive sampling to the $\F\neq\C$ setting discussed in
Section~\ref{fneqc}. 

We borrow heavily from the exposition in~\cite{Dasgupta13}, which simplified somewhat the 
analysis in~\cite{ArthurV07}. However, we encounter some
significant difficulties due to the fact that the proxy
cost function $\sum_{\cli\in\C}\bigl(d(\cli,S)-\beta t_\ell\bigr)^+$ does not satisfy the
traingle inequality, even in the approximate sense considered by~\cite{statman2020k}.%
\footnote{Note that for the proxy cost %used for $\topl$ norms, which is of the form
$\bigl(d(\cli,s)-\tht\bigr)^+$, there is {\em no} finite $\rho$ 
%does not satisfy the $\rho$-approximate triangle
%inequality~\cite{statman2020k}, which requires that 
such that 
$\bigl(d(\cli,s)-\tht\bigr)^+\leq\rho\bigl[\bigl(d(\cli,p)-\tht\bigr)^++\bigl(d(p,s)-\tht\bigr)^+\bigr]$
holds for all $\cli,s,p\in\C$: %for {\em any} finite $\rho$.
%such that holds for all $\cli,s,p\in\C$: 
%Specifically, one could easily 
we could have $d(\cli,s)>\tht$ but $d(\cli,p), d(p,s)\leq\tht$.}   

Recall that $\optSoln=\{\optcen_1,\ldots,\optcen_k\}$ denotes some fixed optimal solution
for the $\topl$ $k$-clustering problem, and $\clustroot_1,\ldots,\clustroot_k$ are the
corresponding optimal clusters, where $\clustroot_q$ is the set of clients assigned to
center $\optcen_q$, for all $q\in[k]$. 
%Let $\OPT=\opt_\ell=\topl\bigl(d(\C,\optsoln)\bigr)$ denote the optimal $\topl$-cost.
As in~\cite{ArthurV07,Dasgupta13}, we divide the optimal clusters
$\clustroot_1,\ldots,\clustroot_k$ into two categories, those that are hit by the current
center-set $S$ (i.e., $S\cap\clust\neq\es$), and the un-hit clusters.
The analysis for $k$-means comprises at a high level two main ingredients: (1) showing
that the expected cost of the clusters that are hit is $O(1)$ times the optimum, and
(2) charging the expected cost of the un-hit clusters to the expected cost of the hit
clusters. For $\topl$ norms, the second portion of the analysis (using the proxy cost)
proceeds similarly, but the first part of the analysis becomes more involved, and we need
to proceed differently to take into account the lack of triangle inequality.

Recall that adaptive sampling for $\topl$-norms proceeds as follows. Let $t_\ell$ be an
estimate of $t^*_\ell=d(\C,\optSoln)^\downarrow_\ell$, the $\ell$-th largest cost incurred
by an optimal solution. 
We add centers iteratively: let $S_0:=\es$, and $S_i$ denote the center-set after $i$
iterations. %so $|S_i|=i$ for all $i\in[k]$. 
%We start with $S_0:=\es$. 
In the first iteration, we pick a point uniformly at random from $\C$ to form $S_1$; in
each subsequent iteration $i$, %letting $S_{i-1}$ denote the 
we sample $s_i\in\C$ with probability proportional to 
$\bigl(d(s_i,S_{i-1})-\nbetaval t_\ell\bigr)^+$ and set $S_i\assign S_{i-1}\cup \{s_{i-1}\}$.
We continue this for $k$ iterations. 
We prove the following.
%%Take $\nbetaval=3$.

\begin{theorem} \label{kiterbnd} \label{finbnd}
Suppose that $t^*_\ell\leq t_\ell\leq\max\bigl\{(1+\ve)t^*_\ell,\frac{\ve\OPT}{\ell}\bigr\}$.
Then $\E{\topl\bigl(d(\C,S_k)\bigr)}\leq O(\log k)\cdot\OPT$.
\end{theorem}

\subparagraph*{Analysis.}
For sets $C,S\sse\C$, and $\tht\in\R_+$, define
$\cost(C,S;\tht):=\sum_{\cli\in C}\bigl(d(\cli,S)-\tht\bigr)^+$; with $\tht=\beta t_\ell$, this
measures the total proxy cost of clients in $C$ with respect to center-set $S$.
We abbreviate $\cost(\{\cli\},S;\tht)$ to $\cost(\cli,S;\tht)$, and $\cost(C,\{\cli\};\tht)$ to
$\cost(C,\cli;\tht)$. 
%So the proxy cost of a cluster $\clustroot_q$ with respect to center-set $S$ is given by
%$\cost(\clustroot_q,S;\beta t_\ell)^+$ 

As in Section~\ref{adsample-topl}, it will be convenient to perform the analysis in terms
of parameters $\al,\beta,\kp,\tau$, where
\begin{equation}
\al=2, \quad\ \ \beta=\nbetaval\geq\al+1, \quad \ \ \kp=2\al+4, \quad\ \ 
\tau=\max\bigl\{\al+\kp+2,\,2\bigl(1+\tfrac{\al(2\al+5)}{\al-1}\bigr)\bigr\}.
\label{toplchoices}
\end{equation}

%The concepts below will be useful, many of which were introduced in
%Section~\ref{adsample-topl}. 
%Some of these notions are with respect to the current
%center-set $S$. 
%which would be the center-set at the end of an iteration of the algorithm.
%
%\begin{enumerate}[label=$\bullet$, topsep=0.2ex, itemsep=0.1ex, leftmargin=*]
%\item Let $\optsSoln=\{\optcen_1,\ldots,\optcen_k\}$, where $\optcen_q$ is the center of
%  $\clustroot_q$. 

%\item $\clustroot_q$ is $\ell$-\emph{good}, if 
%$\cost(\clustroot_q,S;\beta  t_\ell)^+\leq\gamma\cdot\cost(\clustroot_q,\optcen_q;t_\ell)$; 
%otherwise, it is $\ell$-bad.

%\item 
Recall that the radius of a cluster $\clustq$ is
$r_\ell(\clustq) := \frac{\cost(\clustq,\optcen_q;t_\ell)}{|\clustq|}$, and 
the $\ell$-core of $\clustq$ is %defined as  
$\lcore(\clustq):=\{\cli\in \clustq : d(\cli, \optcen_q) \leq \alpha(t_\ell + \rl(\clustq))\}$.     
Define $\num_q:=|\clustq-\lcore(\clustq)|$.
%
%\item 
Say that $\clustq$ is {\em $\ell$-close} with respect to the current center-set $S$
if $d(\optcen_q,S)\leq\kp\max\bigl\{r_\ell(\clustq),t_\ell\bigr\}$; otherwise
$\clustq$ is {\em $\ell$-far}.
%\end{enumerate}

For $i=0,1,\ldots,k$,
let $\hit_i:=\{q\in[k]: S_i\cap\clustq\neq\es\}$, $\unhit_i:=[k]-\hit_i$,
and $\lost_i:=i-|\hit_i|$.
While $\hit_i$ and $\unhit_i$ track the clusters that are hit and not hit respectively
in the first $i$ iterations, $\lost_i$ counts the ``lost'' or ``wasted'' iterations among
the first $i$ iterations, where the algorithm ``lost out'' on hitting a new cluster. 
%denote the indices of the optimal clusters that are hit by the first $i$ centers.
%and $\unhit_i:=[k]-\hit_i$.
Define $\hitcost_i:=\sum_{q\in\hit_i}\cost(\clustq,S_i;\beta t_\ell)$, and
$\Psi_i:=\lost_i\cdot\frac{\sum_{q\in\unhit_i}\cost(\clustq,S_i;\beta t_\ell)}{|\unhit_i|}$.
Note that all of the above quantities are {\em random variables}.
Observe that $\hitcost_0=0$ (since $\hit_0=\es$) and $\Psi_0=\Psi_1=0$ (since
$\lost_0=\lost_1=0$), and 
$\hitcost_k+\Psi_k=\cost(\C,S_k;\beta t_\ell)$, since $\lost_k=|\unhit_k|$.

We will prove that 
$\E{\hitcost_i}\leq O(1)\cdot\bigl(\cost(\C,\optSoln;t_\ell)+\ell t_\ell\bigr)$ 
for all $i\in[k]$ if $t_\ell\geq t^*_\ell$
(Lemma~\ref{hitbnd}), %Corollary~\ref{hitcor})
%O(1)\cdot\sum_{q\in[k]}\Pr[q\in\hit_i]\cdot\bigl(\cost(\clustq,\optcen_q;t_\ell)+\num_q
%t_\ell\bigr)$,
and $\E{\Psi_{i}-\Psi_{i-1}}\leq\frac{\E{\hitcost_{i-1}}}{k-i+1}$ for all $i\in[k]$
(Lemma~\ref{unhitbnd}). 
%Combining the two yields $\E{\hitcost_k+\Psi_k}\leq O(\ln k)\cdot\OPT$.
It follows that $\E{\hitcost_k+\Psi_k}\leq O(\log k)\cdot\bigl(\cost(\C,\optSoln;t_\ell)+\ell t_\ell\bigr)$. 
So if $t_\ell$ is a good estimate of $t^*_\ell$, we obtain that the expected cost of the
solution returned is $O(\log k)\cdot\OPT$, proving Theorem~\ref{finbnd}.

For the first iteration, the following claim will be useful. 

\begin{claim} \label{firstiter}
For any cluster $\clustq$, we have 
$\E{\cost(\clustq,S_1;\beta t_\ell)\,|\,q\in\hit_1}\leq 2\cdot\cost(\clustq,\optcen_q;t_\ell)$.
\end{claim}

\begin{proof}
The first center $s$ is chosen uniformly at random from $\C$. So the expectation in the claim
statement is
\begin{equation*}
%\begin{split}
\frac{1}{|\clustq|}\cdot\sum_{s\in\clustq}\sum_{\cli\in\clustq}\bigl(d(\cli,s)-\beta t_\ell\bigr)^+
 \leq\frac{1}{|\clustq|}\cdot\sum_{s,\cli\in\clustq}
\bigl(d(\cli,\optcen_q)+d(\optcen_q,s)-\beta t_\ell\bigr)^+ %\\
\leq 2\cdot\sum_{\cli\in\clustq}\bigl(d(\cli,\optcen_q)-t_\ell\bigr)^+
%\end{split}
\end{equation*}
where the last inequality is because $\beta\geq 2$.
\end{proof}

%Let $\tau=\max\bigl\{\al+\kp+1,2\bigl(1+\frac{\al(2\al+\beta+3)}{\al-1}\bigr)\bigr\}$.
%Recall that $\optSoln$ 
The bound on $\E{\hitcost_i}$ will follow from the following result.

\begin{restatable}{lemma}{firstnewhit}
%\begin{lemma} 
\label{newhit}
Consider any iteration $i\in[k]$ and any $q\in[k]$. 
%Let $Z$ be the random center drawn in iteration $i$.
Then 
\begin{equation*}
\E[\text{\textup{$Z$ sampled in iteration $i$}}]
{\cost(\clustq,S_{i-1}\cup\{Z\};\beta t_\ell)\,|\,S_{i-1}, \{Z\in\clustq\}}
\leq \tau\cdot\cost(\clustq,\optcen_q;t_\ell)+(\kp-2)\num_q\cdot t_\ell.
\end{equation*}
%\end{lemma}
\end{restatable}

The proof of Lemma~\ref{newhit} is somewhat technical and long, 
so we defer this to the end of the section to avoid detracting the reader.

\begin{lemma} \label{hitbnd}
Suppose $t_\ell\geq t^*_\ell$.
For any $i\in[k]$, we have 
$\E{\hitcost_i}\leq\tau\cdot\cost(\C,\optSoln;t_\ell)+(\kp-2)\ell t_\ell$.
\end{lemma}

\begin{proof}
We use induction on $i$ to show that
$\E{\hitcost_i}\leq\sum_{q\in[k]}\Pr[q\in\hit_i]\cdot
\bigl(\tau\cdot\cost(\clustq,\optcen_q;t_\ell)+(\kp-2)\num_q\cdot t_\ell\bigr)$ for all
$i\in[k]$. 
This immediately yields the lemma because $t_\ell\geq t^*_\ell$ implies that
$\sum_{q\in[k]}\num_q\leq\ell$, since any client
$\cli\in\bigcup_{q\in[k]}\bigl(\clustq-\lcore(\clustq)\bigr)$ incurs assignment cost at
least $t_\ell$ under the optimal solution $\optSoln$.

The base case, when $i=1$, follows from Claim~\ref{firstiter}, which shows that 
$\E{\hitcost_1\,|\,q\in\hit_1}\leq 2\cdot\cost(\clustq,\optcen_q;t_\ell)$ for any
$q\in[k]$. Suppose the above statement holds for $i-1$. Consider index $i$. 
Let us condition on $S_{i-1}$. Note that $\hit_{i-1}$ and $\hitcost_{i-1}$ also get fixed 
%become deterministic quantities 
under this conditioning.
Consider $\E{\hitcost_i-\hitcost_{i-1}\,|\,S_{i-1}}$. If $q\in\hit_{i-1}$, then the
contribution of $\clustq$ to $\hitcost_i-\hitcost_{i-1}$ is non-positive. So we have 
\begin{alignat*}{1}
\Exp&\bigl[\hitcost_i-\hitcost_{i-1}\,|\,S_{i-1}\bigr]
\leq \sum_{q\in\unhit_{i-1}}\Pr[q\in\hit_i\,|\,S_{i-1}]\cdot
\E[Z]{\cost(\clustq,S_{i-1}\cup\{Z\};\beta t_\ell)\,|\,S_{i-1}, \{Z\in\clustq\}} \\
& \leq \sum_{q\in\unhit_{i-1}}\Pr[q\in\hit_i\,|\,S_{i-1}]\cdot
\Bigl(\tau\cdot\cost(\clustq,\optcen_q;t_\ell)+(\kp-2)\num_q\cdot t_\ell\Bigr) 
\tag{using Lemma~\ref{newhit}}
\end{alignat*}
So removing the conditioning on $S_{i-1}$, we obtain that 
\[
\E{\hitcost_i-\hitcost_{i-1}}\leq\sum_{q\in[k]}\Pr[q\in\hit_i-\hit_{i-1}]\cdot
\bigl(\tau\cdot\cost(\clustq,\optcen_q;t_\ell)+(\kp-2)\num_q\cdot t_\ell\bigr).
\]
Combining this with the induction hypothesis for $i-1$, we obtain that
$\E{\hitcost_i}$ is at most $\sum_{q\in[k]}\Pr[q\in\hit_i]\cdot
\bigl(\tau\cdot\cost(\clustq,\optcen_q;t_\ell)+(\kp-2)\num_q\cdot t_\ell\bigr)$. 
\end{proof}

\begin{lemma} \label{unhitbnd} 
For any $i\in[k]$, we have $\E{\Psi_i-\Psi_{i-1}}\leq\frac{\E{\hitcost_{i-1}}}{k-i+1}$. 
\end{lemma}

\begin{proof}
This holds for $i=1$, since $\lost_1=0$.
So consider $i>1$, and condition on $S_{i-1}$. 
Let $Z$ be the random center chosen in iteration $i$, so $S_i=S_{i-1}\cup\{Z\}$.
If $\hit_i=\hit_{i-1}$, then $\lost_i=\lost_{i-1}+1$, and so we have
\begin{equation*}
\begin{split}
\Psi_i-\Psi_{i-1} & =\frac{\sum_{q\in\unhit_{i-1}}\bigl(\lost_i\cdot\cost(\clustq,S_i;\beta t_{\ell})
-\lost_{i-1}\cdot\cost(\clustq,S_{i-1};\beta t_\ell)\bigr)}{|\unhit_{i-1}|} \\
& \leq\frac{\sum_{q\in\unhit_{i-1}}\cost(\clustq,S_{i-1};\beta t_\ell)}{|\unhit_{i-1}|}.
\end{split}
\end{equation*}
Suppose $\hit_i\supsetneq\hit_{i-1}$. We have
\begin{equation*}
\begin{split}
\Exp\Bigl[\sum_{q\in\unhit_i}&\cost(\clustq,S_i;\beta t_\ell)\,|\,S_{i-1},\{\hit_i\supsetneq\hit_{i-1}\}\Bigr]
\\ & \leq \sum_{q\in\unhit_{i-1}}\cost(\clustq,S_{i-1};\beta t_\ell)\cdot
\bigl(1-\Pr[Z\in\clustq\,|\,S_{i-1},\{\hit_i\supsetneq\hit_{i-1}\}]\bigr) \\
& = \sum_{q\in\unhit_{i-1}}\cost(\clustq,S_{i-1};\beta t_\ell)-
\sum_{q\in\unhit_{i-1}}\frac{\cost(\clustq,S_{i-1};\nbetaval t_\ell)}
{\sum_{r\in\unhit_{i-1}}\cost(\clustroot_r,S_{i-1};\nbetaval t_\ell)}\cdot
\cost(\clustq,S_{i-1};\beta t_\ell) \\
& \leq\sum_{q\in\unhit_{i-1}}\cost(\clustq,S_{i-1};\beta t_\ell)-
\sum_{q\in\unhit_{i-1}}\frac{\cost(\clustq,S_{i-1};\nbetaval t_\ell)}{|\unhit_{i-1}|}
%& \leq \sum_{q\in\unhit_{i-1}}\cost(\clustq,S_{i-1};\beta t_\ell)\cdot\frac{|\unhit_{i-1}|-1}{|\unhit_{i-1}|}$,
\end{split}
\end{equation*}
where the last inequality is because $\beta=\nbetaval$ and using the Cauchy-Schwartz
inequality. 
Since $\hit_i\supsetneq\hit_{i-1}$ implies that $|\unhit_i|=|\unhit_{i-1}|-1$ and
$\lost_i=\lost_{i-1}$, the inequalities above imply that  
$\E{\Psi_i\,|\,S_{i-1},\{\hit_i\supsetneq\hit_{i-1}\}}\leq\Psi_{i-1}$.

So we obtain that
\begin{equation*}
\begin{split}
\E{\Psi_i-\Psi_{i-1}\,|\,S_{i-1}} & \leq \Pr[\hit_i=\hit_{i-1}\,|\,S_{i-1}]\cdot
\frac{\sum_{q\in\unhit_{i-1}}\cost(\clustq,S_{i-1};\beta t_\ell)}{|\unhit_{i-1}|} \\
& =\sum_{q\in\hit_{i-1}}\Pr[Z\in\clustq\,|\,S_{i-1}]\cdot
\frac{\sum_{q\in\unhit_{i-1}}\cost(\clustq,S_{i-1};\beta t_\ell)}{|\unhit_{i-1}|} \\
& =\frac{\sum_{q\in\hit_{i-1}}\cost(\clustq,S_{i-1};\nbetaval t_\ell)}
{\cost(\C,S_{i-1};\nbetaval t_\ell)}\cdot
\frac{\sum_{q\in\unhit_{i-1}}\cost(\clustq,S_{i-1};\beta t_\ell)}{|\unhit_{i-1}|} \\
& \leq \frac{\sum_{q\in\hit_{i-1}}\cost(\clustq,S_{i-1};\nbetaval t_\ell)}{|\unhit_{i-1}|}
\leq \frac{\hitcost_{i-1}}{k-i+1}
\end{split}
\end{equation*}
where the last inequality is because $|\unhit_{i-1}|\geq k-i+1$. Removing the conditioning
on $S_{i-1}$ yields the lemma.
\end{proof}

\begin{proofof}{Theorem~\ref{finbnd}}
We can finally put everything together to prove Theorem~\ref{finbnd}.
Recall that $t^*_\ell\leq t_\ell\leq\max\bigl\{(1+\ve)t^*_\ell,\frac{\ve\OPT}{\ell}\bigr\}$.
By Lemmas~\ref{hitbnd} and~\ref{unhitbnd}, we can infer that 
\[
\E{\hitcost_k+\Psi_k}=\E{\hitcost_k}+\sum_{i\in[k]}\E{\Psi_i-\Psi_{i-1}}\leq 
(1+H_k)\bigl(\tau\cdot \cost(\C,\optSoln;t_\ell)+(\kp-2)\ell t_\ell\bigr).
\]
So we have
\begin{equation*}
\begin{split}
\E{\topl\bigl(d(\C,S_k)\bigr)} & \leq\ell\cdot\beta t_\ell+\E{\cost(\C,S_k;\beta t_\ell)} 
= \ell\cdot\beta t_\ell+\E{\hitcost_k+\Psi_k} \\
& \leq \ell\cdot\beta t_\ell+(1+H_k)\tau\cdot\cost(\C,\optSoln;t_\ell)+
(1+H_k)(\kp-2)\ell t_\ell \\
%& \leq \bigl(\beta+(1+H_k)(\kp-2)\bigr)\max\bigl\{(1+\ve)\ell t^*_\ell,\ve\OPT\bigr\}+
%(1+H_k)\tau\cdot\cost(\C,\optSoln;t^*_\ell) \\
& \leq O(\log k)\cdot\bigl(\max\bigl\{(1+\ve)\ell t^*_\ell,\ve\OPT\bigr\}+
\cost(\C,\optSoln;t^*_\ell)\bigr) \\
& = O(\log k)\cdot\bigl(\ell t^*_\ell+\cost(\C,\optSoln;t^*_\ell)\bigr)
+O(\log k)\cdot\ve\OPT \leq O(\log k)\cdot\OPT. \qedhere
\end{split}
\end{equation*}
\end{proofof}

Finally, we prove Lemma~\ref{newhit}. We restate the lemma for convenience.

\firstnewhit*

\begin{proof} %of}{Lemma~\ref{newhit}}
When $i=1$, this is implied by Claim~\ref{firstiter}, so suppose $i>1$.
We consider the cases where $\clustq$ is $\ell$-close and $\ell$-far separately.
We abbreviate $\rl(\clustq)$ to $\rl$ to avoid notational clutter.

\medskip
\noindent
{\bf\boldmath {$\clustq$ is $\ell$-close.}}\
%The idea here is that 
If $Z\in\lcore(\clustq)$, then Claim~\ref{coreisgood} implies that 
$\cost(\clustq,Z;\beta t_\ell)\leq(\al+1)\cost(\clustq,\optcen_q;t_\ell)$, since 
Claim~\ref{coreisgood} shows that when $\beta\geq\al+1$, for any $s\in\lcore(\clustq)$, we
have $\cost(\clustq,s;\beta t_\ell)\leq(\al+1)\cost(\clustq,\optcen_q;t_\ell)$.
So we have
\begin{alignat}{1}
{\textstyle \Exp_Z}\bigl[\cost(\clustq,S_{i-1}\cup\{Z\};\beta t_\ell)\,&|\,S_{i-1}, \{Z\in\clustq\}\bigr] 
\leq\frac{\Pr[Z\in\lcore(\clustq)]}{\Pr[Z\in\clustq]}\cdot(\al+1)\cost(\clustq,\optcen_q;t_\ell)
\notag \\ 
& +\frac{\Pr[Z\in\clustq-\lcore(\clustq)]\cdot\cost(\clustq,S_{i-1};\beta t_\ell)}{\Pr[Z\in\clustq]}
\label{newhit-ineq1}
\end{alignat}
where all probabilities are in the conditional space, where we condition on $S_{i-1}$.
For $Z\in\clustq-\lcore(\clustq)$, we utilize that $\clustq$
is $\ell$-close to charge the expected cost to $\cost(\clustq,\optcen_q;t_\ell)$ and
$\num_q\cdot\max\{\rl,t_\ell\}$. (Recall that $\num_q=|\clustq-\lcore(\clustq)|$.)

For any set $C\sse\C$, we have 
$\Pr[Z\in C]=\cost(C,S_{i-1};\nbetaval t_\ell)/\cost(\C,S_{i-1};\nbetaval t_\ell)$, and we
have $\cost(\clustq,S_{i-1};\beta t_\ell)\leq\cost(\clustq,S_{i-1};\nbetaval t_\ell)$.
So the second term on the RHS of \eqref{newhit-ineq1} is at most
\begin{equation*}
\begin{split}
\cost(\clustq-\lcore(\clustq),S_{i-1};\nbetaval t_\ell) & \leq
\sum_{\cli\in\clustq-\lcore(\clustq)}\bigl(d(\cli,\optcen_q)-t_\ell\bigr)^+
%+\num_q\cdot\bigl(d(\optcen_q,S_{i-1})-(\nbetaval-1)t_\ell\bigr)^+
+\num_q\cdot\bigl(d(\optcen_q,S_{i-1})-2t_\ell\bigr)^+ \\
%\leq\cost(\clustq,\optcen_q;t_\ell)+\num_q(\kp-\nbetaval+1)\max\{\rl,t_\ell\}
& \leq\cost(\clustq,\optcen_q;t_\ell)+\num_q\bigl(\kp\rl+(\kp-2)t_\ell\bigr)
\end{split}
\end{equation*}
where the last inequality is because 
$d(\optcen_q,S_{i-1})\leq\kp\max\{\rl,t_\ell\}\leq\kp(\rl+t_\ell)$, 
as $\clustq$ is $\ell$-close. 
%Now if $\rl\geq t_\ell$, then 
Since $\num_q\cdot\kp\rl$ is at most %$\num_q(\kp-\nbetaval+1)\rl$ is at most
%$(\kp-\nbetaval+1)\cost(\clustq,\optcen_q;t_\ell)$, so we have
$\kp\cdot\cost(\clustq,\optcen_q;t_\ell)$, we can say that
\[
\cost(\clustq-\lcore(\clustq),S_{i-1};\nbetaval t_\ell)
%\leq(\kp-\nbetaval+2)\cost(\clustq,\optcen_q;t_\ell)+(\kp-\nbetaval+1)\num_q\cdot t_\ell$.
\leq(\kp+1)\cdot\cost(\clustq,\optcen_q;t_\ell)+(\kp-2)\num_q\cdot t_\ell.
\]
Plugging this in \eqref{newhit-ineq1}, when $\clustq$ is $\ell$-close, we obtain that 
\begin{equation}
\E[Z]{\cost(\clustq,S_{i-1}\cup\{Z\};\beta t_\ell)\,|\,S_{i-1}, \{Z\in\clustq\}}
%\leq(\al+\kp-\nbetaval+3)\cost(\clustq,\optcen_q;t_\ell)+(\kp-\nbetaval+1)\num_q\cdot t_\ell.
\leq(\al+\kp+2)\cost(\clustq,\optcen_q;t_\ell)+(\kp-2)\num_q\cdot t_\ell.
\label{newhit-ineq2}
\end{equation}

\medskip
\noindent
{\bf\boldmath {$\clustq$ is $\ell$-far.}}\
%The expectation in the lemma statement is at most
We have
\begin{equation}
\begin{split}
{\textstyle \Exp_Z}\bigl[\cost(\clustq,S_{i-1}&\cup\{Z\};\beta t_\ell)\,|\,S_{i-1}, \{Z\in\clustq\}\bigr]
\\ & \leq\sum_{s\in\clustq}
\frac{\cost(s,S_{i-1};\nbetaval t_\ell)}{\cost(\clustq,S_{i-1};\nbetaval t_\ell)}\cdot
\min\Bigl\{\cost(\clustq,S_{i-1};\beta t_\ell),\cost(\clustq,s;\beta t_\ell)\Bigr\}. 
\end{split}
\label{newhit-ineq3}
\end{equation}
We have $\cost(s,S_{i-1};\nbetaval t_\ell)\leq
%\bigl(d(s,\optcen_q)-t_\ell\bigr)^+\bigl(d(\cli,\optcen_q)-t_\ell\bigr)^++\cost(\cli,S_{i-1};(\nbetaval-2)t_\ell)$
\cost(s,\optcen_q;t_\ell)+\cost(\cli,\optcen_q;t_\ell)+\cost(\cli,S_{i-1};t_\ell)$
for any $\cli\in\clustq$. Averaging this inequality over all $\cli\in\clustq$, we obtain that
\begin{equation*}
\cost(s,S_{i-1};\nbetaval t_\ell)
\leq\cost(s,\optcen_q;t_\ell)+\frac{\cost(\clustq,\optcen_q;t_\ell)}{|\clustq|}+
%\frac{\cost(\clustq,S_{i-1};(\nbetaval-2)t_\ell)}{|\clustq|}.
\frac{\cost(\clustq,S_{i-1};t_\ell)}{|\clustq|}.
\end{equation*}
Plugging this in \eqref{newhit-ineq3}, and simplifying, we obtain that the expectation in
the lemma statement is at most
\begin{alignat}{1}
& \frac{1}{\cost(\clustq,S_{i-1}\cdot\nbetaval t_\ell)}\cdot
\sum_{s\in\clustq}\biggl(\cost(s,\optcen_q;t_\ell)\cdot\cost(\clustq,S_{i-1};\beta t_\ell)
\notag \\ & \qquad \quad 
+\frac{\cost(\clustq,\optcen_q;t_\ell)}{|\clustq|}\cdot\cost(\clustq,S_{i-1};\beta t_\ell)
%\frac{\cost(\clustq,S_{i-1};(\nbetaval-2)t_\ell)}{|\clustq|}\cdot\cost(\clustq,s;\beta t_\ell)}
+\frac{\cost(\clustq,S_{i-1};t_\ell)}{|\clustq|}\cdot\cost(\clustq,s;\beta t_\ell)\biggr)
\notag \\
& \leq 2\cdot\cost(\clustq,\optcen_q;t_\ell)+
%\frac{\cost(\clustq,S_{i-1};(\nbetaval-2)t_\ell)}{\cost(\clustq,S_{i-1};\beta t_\ell)}\cdot
\frac{\cost(\clustq,S_{i-1};t_\ell)}{\cost(\clustq,S_{i-1};\nbetaval t_\ell)}\cdot
\frac{\sum_{s\in\clustq}\cost(\clustq,s;\beta t_\ell)}{|\clustq|}.
\label{newhit-ineq4}
\end{alignat}
The quantity $\frac{\sum_{s\in\clustq}\cost(\clustq,s;\beta t_\ell)}{|\clustq|}$ in the
final term on \eqref{newhit-ineq4} is at most $2\cdot\cost(\clustq,\optcen_q;t_\ell)$, as
shown in Claim~\ref{firstiter}. So we proceed to bound
%$\frac{\cost(\clustq,S_{i-1};(\nbetaval-2)t_\ell)}{\cost(\clustq,S_{i-1};\nbetaval t_\ell)}$ 
$\frac{\cost(\clustq,S_{i-1};t_\ell)}{\cost(\clustq,S_{i-1};\nbetaval t_\ell)}$
by $O(1)$, which will finish the proof of this case, and hence the lemma.

We have $\cost(\clustq,S_{i-1};t_\ell)\leq
\sum_{\cli\in\clustq}\bigl(d(\cli,\optcen_q)-t_\ell\bigr)^++|\clustq|\cdot d(\optcen_q,S_{i-1})
=|\clustq|\cdot\bigl(\rl+d(\optcen_q,S_{i-1})\bigr)$.
So since $\clustq$ is $\ell$-far, we have
$\cost(\clustq,S_{i-1};t_\ell)\leq \bigl(1+\frac{1}{\kp}\bigr)\cdot|\clustq|\cdot d(\optcen_q,S_{i-1})$.
We lower bound $cost(\clustq,S_{i-1};\nbetaval t_\ell)$ by
$\cost(\lcore(\clustq),S_{i-1};\nbetaval t_\ell)$, which is at least 
\begin{equation*}
\begin{split}
\sum_{\cli\in\lcore(\clustq)}\bigl(d(\optcen_q,S_{i-1})-d(\cli,\optcen_q)-\nbetaval t_\ell\bigr)^+
& \geq |\lcore(\clustq)|\cdot\bigl(d(\optcen_q,S_{i-1})-\al\rl-\al t_\ell-\nbetaval t_\ell\bigr)
\\ & \geq |\lcore(\clustq)|\cdot d(\optcen_q,S_{i-1})\cdot\Bigl(1-\tfrac{2\al+\nbetaval}{\kp}\Bigr).
\end{split}
\end{equation*}
The first inequality above is because $\cli\in\lcore(\clustq)$, and the second is because
$\clustq$ is $\ell$-far. Combining these bounds, and since
$|\lcore(\clustq)|\geq\bigl(1-\frac{1}{\al}\bigr)|\clustq|$,  
we obtain that
$\frac{\cost(\clustq,S_{i-1};t_\ell)}{\cost(\clustq,S_{i-1};\beta t_\ell)}
\leq\frac{\al}{\al-1}\cdot\frac{\kp+1}{\kp-2\al-\nbetaval}
\leq\frac{\al(2\al+5)}{\al-1}$, where the last inequality is because 
$\frac{\kp+1}{\kp-2\al-\nbetaval}$ is a decreasing function of $\kp$ and 
$\kp\geq 2\al+4$. 
%Take $\kp\geq 2\al+\beta+1

Plugging these bounds in \eqref{newhit-ineq4}, when $\clustq$ is $\ell$-far, we obtain
that 
\begin{equation}
\E[Z]{\cost(\clustq,S_{i-1}\cup\{Z\};\beta t_\ell)\,|\,S_{i-1}, \{Z\in\clustq\}}
\leq 2\biggl(1+\frac{\al(2\al+5)}{\al-1}\biggr)\cdot\cost(\clustq;\optcen_q;t_\ell).
\label{newhit-ineq5}
\end{equation}

%\smallskip
The lemma now follows from \eqref{newhit-ineq2} and \eqref{newhit-ineq5}, and the
definition of $\tau$ (see \eqref{toplchoices}).
%\hfill \qed
\end{proof} %of}

\subsection{The setting \boldmath $\F \ne \C$} \label{fneqc}
First, note that the adaptive-sampling algorithm and its analysis extend
%notational clarifications, 
to the setting $\F\supseteq\C$: we still sample a client in step~\ref{sample}, and open a
center at this client in step~\ref{fopen}. Given this, we can proceed in two ways. 
One standard idea that works for many $k$-clustering problems is to simply 
``move'' each client to the center in $\F$ closest to it. We then obtain an instance where
the clients are located at some subset of $\F$ (with possibly multiple co-located
clients), so we can use the adaptive-sampling algorithm from Section~\ref{adsample}. 
This works because the loss in $f$-cost due to moving clients can be bounded in terms of
$\OPT$.

For $S\sse\F$, let $d'(\cli,S)$ denote the assignment-cost of $\cli$ under $S$ for the new
instance, i.e., when $\cli$ has been moved from its original location to the center in
$\F$ closest to it. So we have $|d'(\cli,S)-d(\cli,S)|\leq d(\cli,\F)$.
%$\mu_{\cli}:=d(\cli,\F)$, and 
So if $\OPT'$ is the optimal $f$-cost for the new instance, 
%where we have moved clients to their closest centers in $\F$, 
we have $\OPT'\leq 2\cdot\OPT$: 
if $\optsoln$ is an optimum solution for the original problem, 
then we have $d'(\cli,\optsoln)\leq d(\cli,\optsoln)+d(\cli,\F)$ for every client $\cli$,
so $f\bigl(d'(\C,\optsoln)\bigr)\leq\opt+f\bigl(d(\C,\F)\bigr)$ and
$f\bigl(d(\C,F)\bigr)\leq\OPT$. 
%then considering this
%solution for the new instance, we have that the assignment-cost of a client $\cli$ in the
%new instance is at most $d(\cli,\F)+d(\cli,\optsoln)\leq 2\cdot d(\cli,\optsoln)$. 
%Therefore, the $f$-cost of this solution for the new instance is at most $2\cdot\OPT$. 
Therefore, if $S\sse\F$ is a $\rho$-approximate solution for the new instance, then
we have $d(\cli,S)\leq d(\cli,\F)+d'(\cli,S)$, so we have 
$f\bigl(d(\C,S)\bigr)\leq f\bigl(d(\C,\F)\bigr)+\rho\cdot\OPT'\leq(2\rho+1)\OPT$.
%treating $S$ as a solution for the original instance, we obtain that the assignment cost
%of $\cli$ when it is ``moved back'' to its original location is at most 

\medskip
%The other, 
A cleaner approach, which avoids a factor-$2$ loss in approximation, is to
directly modify adaptive sampling in the 
following simple fashion. In step~\ref{sample}, we still sample a client $s_i\in\C$ with
probability proportional to $\bigl(d(s_i,S_{i-1})-\beta t_{\ell_{\max}}\bigr)^+$ (for a suitable
$\beta$ value); but in step~\ref{fopen}, we now add to $S_i$ the center in $\F$ closest to $s_i$.
This modification was made for $\topl$-norms in~\cite{aaai25b}, and, exactly as with the 
case $\F=\C$, we can rely on their analysis 
%(specifically the analog of Lemma~\ref{choosecore}) 
to extend things to general monotone, symmetric norms.
%To elaborate, 
Suppose that $\al,\beta,\gm,\rho,\tau$ satisfy %the following inequalities: 
\begin{equation}
  \beta = \gm = 2\alpha +1, \quad
  \rho\geq 2(\al+2\beta+3)+\gm, \quad %\kp\geq 2\al+\beta+2, \\
  %1-\tfrac{\gamma }{\rho}\geq 2\cdot\tfrac{\kappa + \beta }{\rho} \\
  %\kappa  & \geq 2\alpha + \beta  +2, \qquad
  \Bigl(1-\tfrac{\gamma}{\rho}\Bigr)\cdot\tfrac{\alpha-1}{2\alpha(\al+\beta+3)}\geq\frac{1}{\tau}.
\label{genchoices}
\end{equation}
In particular, we can take $\al=2$, $\beta=\gm=5$, $\tau=\rho=45$.
These inequalities imply %\eqref{abkchoices}, as also 
the inequalities (5)
in~\cite{aaai25b} (by taking $\kp=\al+\beta+3$), which are used in~\cite{aaai25b} to
analyze adaptive sampling for $\topl$-norms when $\F\neq\C$.
With this choice of parameters, \cite{aaai25b} show that almost
the entire analysis of adaptive sampling for $\topl$-norms in the setting $\F=\C$ 
%from Section~\ref{adsample-topl} 
continues to hold. In particular, with the same definitions for $\ell$-good, $\ell$-bad
clusters, and $\ell$-core (and $\ell$-supercore) of a cluster, Claim~\ref{allgood} 
and Lemma~\ref{choosecore} continue to hold. The only portion of the analysis that we need
to modify is Claim~\ref{coreisgood} showing that if we sample a client from the
$\ell$-core of a cluster then that cluster becomes $\ell$-good. 

\begin{claim} \label{gen-coreisgood}
Consider a cluster $\clust$ and let $S$ be the current center-set. 
Suppose there is some client $s\in\lcore(\clust)$ such that $S$ contains the center in
$\F$ closest to $s$ (so $d(s,S)=d(s,\F)$).
Then $\clust$ is $\ell$-good (and hence remains $\ell$-good throughout). 
\end{claim}

\begin{proof} %{Claim \ref{gen-coreisgood}}
Let $a\in S$ be such that $d(s,a)=d(s,\F)$. Then
$\sum_{\cli\in\clust}(d(\cli,S)-\beta t_\ell)^+\leq\sum_{\cli\in\clust}(d(\cli,a)-\beta t_\ell)^+$, 
which (since $\beta\geq 2\al+1$) is at most 
\begin{equation*}
\begin{split}
\sum_{\cli\in\clust}\bigl((d(\cli,\optcen_q)-t_\ell)^++(d(s,\optcen_q)+d(s,a)-2\al\cdot t_\ell)^+\bigr) 
& \leq |\clust|\cdot\bigl(r_\ell(C^*_q)+(2d(s,\optcen_q) - 2\al\cdot t_\ell)\bigr)^+. \\
& \leq |\clust|\bigl(r_\ell(C^*_q)+2\al\cdot r_\ell(C^*_q)\bigr).
\end{split}
\end{equation*}
The first inequality follows from the triangle inequality, and since $(y+z)^+\leq y^++z^+$;
the second follows from the definition of $r_\ell$, and since $d(s,a)\leq d(s,\optcen_q)$;
The third is because $s\in\core(C^*_q)$. Since $\gm\geq 2\al+1$, this
shows that $\clust$ is $\ell$-good.
\end{proof}

Given the above claim, we obtain the following guarantee for the $\F\neq\C$ setting by the
exact same arguments as in Section~\ref{adsample-gen}. 

\begin{theorem} \label{adsample-genthm}
Let $\vec{t}$ be a non-increasing vector %= (t_1, \ldots, t_n) \in \R^{|\pos|}$ 
such that $t^*_\ell \leq t_\ell \leq\max\{(1 + \veps) t^*_\ell, \frac{\veps\opt}{n}\}$ 
for all $\ell \in \pos$. 
%where $\ve\leq 1$.
%Let $S$ be the set of centers opened by Algorithm \ref{alg:adsample-min-norm} and 
The above modification of Algorithm~\ref{alg:adsample-min-norm} returns a set $S$ of
$O(k)$ centers satisfying 
$f\bigl(d(\C, S)\bigr)\le 45(1 + \varepsilon)^3(1 + \delta)^2\cdot \opt$ 
with constant probability.   
\end{theorem}

Complementing this with
Lemma~\ref{polyset}, we again obtain an $\bigl(O(1),O(1)\bigr)$-bicriteria algorithm for
minimum-norm $k$-clustering when $\F\neq\C$, by running the adaptive-sampling algorithm for
each vector $\vec{t}$ in the set $\T$ output by Lemma~\ref{polyset}, and returning the
best solution found.

\subsection{Faster \boldmath $O(1)$-approximation by sparsifying data} \label{trueapx}
We show that by using our algorithm in conjunction with the (true) $O(1)$-approximation
algorithm for minimum-norm $k$-clustering in~\cite{ChakrabartyS19}, referred to as \csalg
from now on, we can obtain a faster $O(1)$-approximation algorithm. We note that this way  
of obtaining running-time savings was also outlined by~\cite{aggarwal_adaptive_2009} for
the $k$-means problem. 

We consider the $\F=\C$ setting for simplicity (which is the setting also considered
in~\cite{ChakrabartyS19}). 
Let $\post=\pos_{n,1}$, i.e., $\post$ consists of all powers of $2$ up to $n$ (possibly
including $n$).  
%but the same ideas carry over to the $\F\neq\C$ setting 
%The algorithm in~\cite{ChakrabartyS19}, which we will refer to as \csalg from now on, 
Algorithm \csalg %(from~\cite{ChakrabartyS19}) 
also requires knowing the right threshold
vector $\vec{t}$ satisfying the conditions in Theorem~\ref{adsample-thm}. As shown
earlier, we can find a set $\T$ containing such a vector (Lemma~\ref{polyset}). 
For an instance with $n$ clients and a fixed $\vec{t}\in\T$, \csalg
solves an LP, which is the standard $k$-median LP (with $O(n^2)$ variables and
constraints) augmented with $O(n)$ additional constraints,% 
\footnote{The $O(n)$ additional constraints are constraints (OCl-4)
in~\cite{ChakrabartyS19}. Although, it seems like there are quadratically many such
constraints, a closer inspection shows that one needs to consider at most one such
constraint for each client.}
and rounds its optimal solution.
To elaborate, \csalg considers the setting wherein the norm $f$ is the maximum of a
collection 
of ordered norms specifed in the input, and their LP has constraints bounding the (proxy)
cost under each ordered norm in this collection. 
Given Theorem~\ref{thm:majorization}, one
can approximate $f$ by specifying $\topl$-norm budgets for all $\ell\in\post$ 
(obtained from a threshold vector $\vt\in\Rp^{\post}$), so this results in $O(\log n)$
such constraints. %that bound the $\topl$-proxy cost for each $\ell\in\post$.}
In the sequel, we refer to this LP %(that depends on the threshold vector) 
as the augmented $k$-median LP, which has $O(n^2)$ variables and $O(n^2)$ constraints. 
For a given $\vec{t}$, the running time of \csalg is dominated by the time needed to solve 
this augmented $k$-median LP.%
\footnote{The running time of the LP-rounding procedure in \csalg is not explicitly stated
in~\cite{ChakrabartyS19}, but from its description, one can infer that it is $O(n)+\poly(k)$.}
Thus, the overall running time of \csalg, 
%which is dominated by the time needed to solve these LPs 
is $O\bigl(|\T|\cdot\lptime(n)\bigr)$, where $\lptime$ is the time
needed to solve their augmented $k$-median LP on an instance with $n$ clients.

For the faster algorithm, we first use Algorithm~\ref{adsample-alg} to obtain an
$(\al,\rho)$-approximate solution $S$, where $\al,\rho=O(1)$, with constant
probability. This involves running 
Algorithm~\ref{adsample-alg} for each $\vec{t}$ in the set $\T$, and returning the best
solution found. A simple implementation of Algorithm~\ref{adsample-alg} takes $O(nk)$
time, so we can find $S$ in $O(|\T|nk)$ time. Now, we consider the instance where we move
each client to its closest center in $S$, and are allowed to open centers only at
locations in $S$. We run \csalg on the resulting instance, which we call the 
{\em \sparse instance} %(obtained from $S$), 
where, by design, there are at most $|S|\leq\al k$ distinct client locations 
%clients are located at some subset of $S$ 
and we may have co-located clients. %We run \csalg on this \sparse instance
%where clients are located at some subset of $S$ (and we can have co-located clients), and
%we may open centers at locations in $S$.
(While the algorithm in~\cite{ChakrabartyS19} is stated for the setting $\F=\C$, it works
as is for the setting $\F\supseteq\C$ and opens centers at a subset of $\F$.) 
Lemma~\ref{spapprox} shows that moving to the \sparse instance incurs only 
an $O(1)$-factor loss in approximation. %(Lemma~\ref{spapprox}).
%a good solution to this sparsified instance yields a good solution for the original instance 
Since $|S|\leq\al k$, running \csalg on the \sparse instance %with vector $\tsparse$ 
takes time $O\bigl(|\T'|\cdot\lptime(\al k)\bigr)$, where $\T'$ is a set containing a suitable
threshold vector for the \sparse instance (which still has $n$ clients). 
%(Note that the \sparse instance still has $n$ clients.)
So the total running time is $O\bigl(|\T|nk+|\T'|\cdot\lptime(\al k)\bigr)$. 
This is better than the earlier running time, since solving LPs takes
super-quadratic time \cite{JiangSWZ21}), but %notably, 
note also that the time required to solve LPs now only affects the dependence of the
runtime on $k$.  
%For instance, if 

Furthermore, we show in Lemma~\ref{npoly} that one can refine Lemma~\ref{polyset}
so that, for an instance with $n$ clients, one can identify a set of size
$O(n\log n)$ %a set $\T$ with $|\T|=O(n\log n)$ 
that contains a threshold vector $\vec{t}$ such that: (a) running a slightly modified version of
Algorithm~\ref{adsample-alg} with $\vec{t}$ yields an
$\bigl(O(1),O(1)\bigr)$-approximation; and (b) running \csalg with $\vt$ yields an
$O(1)$-approximate solution.
Thus, $|\T|$ and $|\T'|$ can be bounded by $O(n\log n)$, so 
%This requires subtly changing the interpretation of $t_\ell$, from being an estimate of
%$d(\C,\optSoln)^{\down}_\ell$, to an estimate of $\topl\bigl(d(\C,\optSoln)\bigr)/\ell$. 
%
%In particular, with $\dt=\ve=1$, one can tighten Lemma~\ref{toplbnds} and
%Lemma~\ref{polyset} to obtain $|\mathcal{T}| = O(n)$. 
Lemmas~\ref{npoly} and~\ref{spapprox} show that this sparsify-and-solve approach
yields an $O(1)$-approximation (with constant probability) in 
$O\bigl(n\log n\cdot(nk+\lptime(O(k)))\bigr)$ time. 
%and the dependence on $n$ remains $O(n^2 k)$ -- 
In particular, for constant $k$, the overall runtime is now roughly \emph{quadratic} in $n$. In
contrast, even assuming ({\em very} generously) that we have a \emph{linear-time}
algorithm for (approximately) solving LPs, we would obtain $\lptime(n)=O(n^2)$ (since the
LP has $n^2$ variables), and so running \csalg on the original instance would require
$\Omega(n^3\log n)$ time.    
%which is better than the earlier running time (since solving LPs takes
%super-quadratic time).  
%even a {\em very} generous estimate of the running time of \csalg would place its on the
%original instance would place its running 

%Finally, we can establish the following approximation guarantee, taking the same approach as 
%in~\cite{guha_focs2000}. %Recall that %Algorithm~\ref{adsample-alg} returns an

\begin{lemma}\label{sparsify-data-approx} \label{spapprox}
Let $S$ be an $(\al,\rho)$-approximate solution for the original instance.
Let $\tS\sse S$ be a $\lambda$-approximate solution for the new instance. Then the $f$-cost of
$\tS$ for the original instance is at most $(2\ld+2(\ld+1)\rho)\cdot\OPT$.
\end{lemma}

\begin{proof} %{Lemma~\ref{spapprox}} %{Claim~\ref{sparsify-data-approx}}
%Recall that $S$ is an $(\al,\rho)$-approximate solution for the original instance. 
We follow a similar approach as in~\cite{guha_focs2000}.
Let $\OPT'$ denote the optimum for the modified
instance where we have moved clients, and can open centers only at points in $S$. We claim
that $\OPT'\leq 2\cdot\OPT+2\cdot f\bigl(d(\C,S)\bigr)$. Consider an
optimal solution $\optsoln$, and map each center in $\optsoln$ to its closest point in
$S$. Let $S^*\sse S$ denote the resulting subset of (at most) $k$ centers from $S$. 
Consider any client $\cli\in\C$, which is assigned to center $\optcen\in\optsoln$, and
which was moved to location $a\in S$. %Note that $d(\optcen,S)\leq d(\optcen,a)$.
We can bound the assignment cost of $\cli$ under $S^*$
for the new instance as follows: consider moving $\cli$ from $a$ to $\cli$, then to $\optcen$,
and then moving from $\optcen$ to the point in $S$ closest to $\optcen$. This yields an
assignment cost of at most 
\begin{equation*}
\begin{split}
d(\cli, S)+d(\cli,\optcen)+d(\optcen, S) & \leq d(\cli, S)+d(\cli,\optcen)+d(\optcen, a) \\
& \leq d(\cli, S)+d(\cli,\optcen)+d(\optcen,\cli)+d(\cli, a) = 2\cdot d(\cli, S)+2\cdot d(\cli,\optcen).
\end{split}
\end{equation*}
So the $f$-cost of center-set $S^*$ for the new instance is at most 
$2\cdot\OPT+2\cdot f\bigl(d(\C,S)\bigr)$.

%Now suppose we obtain a $\ld$-approximate solution $\tS\sse S$ for the new
%instance. 
Recall that $\tS\sse S$ is a $\ld$-approximate solution for the new instance.
Then, we can bound the assignment cost of a client under $\tS$ for the original
instance by simply ``moving back'' the client from its location in $S$ to its original
location. This moving-back step incurs an additional $f$-cost of
$f\bigl(d(\C,S)\bigr)$. So the $f$-cost of $\tS$ for the original instance is at most 
\[
\ld\cdot\OPT'+f\bigl(d(\C,S)\bigr)\leq 2\ld\cdot\OPT+(2\ld+1)\cdot f\bigl(d(\C,S)\bigr)
\leq (2\ld+(2\ld+1)\rho)\cdot\OPT. \qedhere
\]
\end{proof}

\begin{lemma} \label{npoly}
Let $\post=\pos_{n,1}$. 
%i.e., $\post$ consists of all powers of $2$ up to $n$ (possibly
%including $n$).  
%
\begin{enumerate}[(a)] %[label=(\alph*), topsep=0.2ex, itemsep=0.1ex, leftmargin=*]
\item We can obtain a set $\T\sse\R_+^{\post}$ with 
$|\T|\leq O(n\log n)$ such that $\T$ contains a vector $\vt$ with 
$\topl\bigl(d(\C,\optsoln)\bigr)\leq \ell t_\ell\leq 2\cdot\topl\bigl(d(\C,\optsoln)\bigr)$
for all $\ell\in\post$.

\item Running Algorithm~\ref{adsample-alg} with threshold vector $\vt$, but changing
step~\ref{blestim} so as to define $B_\ell:=\ell t_\ell$ for all $\ell\in\post$, yields an 
$\bigl(O(1),O(1)\bigr)$-approximate solution, with constant probability. 

\item Running \csalg with threshold vector $\vt$ yields an $O(1)$-approximate solution.
\end{enumerate}
\end{lemma}

\begin{proof}
%We prove part (a) here, and defer the proofs of parts (b) and (c) to
%Appendix~\ref{append-trueapx}. 
Define $\htt_\ell={\topl\bigl(d(\C,\optsoln)\bigr)}/{\ell}$ for $\ell\in\post$. The
benefit of considering the target vector $\htt$ (as opposed to the vector
$\bigl\{d(\C,\optsoln)^{\down}_\ell\bigr\}_{\ell\in\post}$), is that $\htt$ is not
only a non-increasing vector, but consecutive coordinates decrease by a factor of at most
$2$: for $\ell\in\post, \ell>1$, we have
$\frac{\htt_{\ell/2}}{2}\leq\htt_\ell\leq \htt_{\ell/2}$. This is helpful because then the
nearest powers-of-$2$ vector $\vt$ that coordinate-wise overestimates $\htt$ can be
described by specifying $t_1$ and the subset of $\post$ where the coordinate values drop
by a factor of $2$.

As in the proof of Lemma~\ref{polyset}, we can estimate $\htt_1$ within an $O(n)$-factor as
follows. %a power of $2$. More precisely, 
Let $S^*_1$ be the output of the $2$-approximation
algorithm for $k$-center due to~\cite{Gonzalez}, and let $\ub$ be the smallest power of
$2$ that is at least $n\cdot d(\C, S^*_1)^{\down}_1$. Then, the proof of
Lemma~\ref{polyset} shows that $\frac{\ub}{4n}\leq\htt_1\leq\ub$. Define 
\begin{equation*}
\begin{split}
\T=\Bigl\{v\in\R_+^{\post}:\ \ & v_\ell\text{ is a power of $2$}\ \ \ \forall \ell\in\post; \\
& v_1\in\bigl[\tfrac{\ub}{4n},\ub\bigr], \quad
\tfrac{v_{\ell/2}}{2}\leq v_\ell\leq v_{\ell/2}\ \ \ \forall \ell\in\post,\ \ell>1\Bigr\}.
\end{split}
\end{equation*}
We show that $\T$ satisfies the properties stated in part (a). 
Let $\vt\in\R_+^{\post}$ be such that $t_\ell$ is the smallest power of $2$ that is at least
$\htt_\ell$, for all $\ell\in\post$. We have $t_1\leq\ub$ and 
$\htt_\ell\leq t_\ell<2\htt_\ell$ for all $\ell\in\post$ by design, which also implies that 
$\frac{t_{\ell/2}}{2}\leq t_\ell\leq t_{\ell/2}$ for all $\ell\in\post-\{1\}$; so
$\vt\in\T$. 
Any vector $v\in\T$ is {\em uniquely} determined by $v_1$ and the subset
$S\sse\post-\{1\}$ for which $v_\ell=\frac{v_{\ell/2}}{2}$. %This mapping is one-to-one. 
There are at most $n$ subsets of $\post-\{1\}$, as $|\post|=O(\log n)$, and
$O(\log n)$ choices for $v_1$, so we have $|\T|=O(n\log n)$. %This proves part (a).

\medskip
For part (b), let $\al=2$, $\beta=\gm=3$, $\tau=\rho=35$ as in
Section~\ref{adsample-topl}. We only need to prove the analogue of Claim~\ref{allgood}
showing that if all clusters are $\ell$-good (when considering the input vector $\vt$ from
part (a) in Algorithm~\ref{adsample-alg}), for some index $\ell\in\post$, then the
$\topl$-cost of the solution is at most $\rho\cdot\topl\bigl(d(\C,\optsoln)\bigr)$,
and so $\ell\notin I$. Given this, the proof of Theorem 3.7 shows that at the end of 
$N=\ceil{\tau(k+\sqrt{k})}$ iterations, we have $I=\es$ with constant probability. 
This implies that
$\topl\bigl(d(\C,S)\bigr)\leq\rho(1+\ve)\cdot\topl\bigl(d(\C,\optsoln)\bigr)$ for all
$\ell\in\post$ and hence, $f\bigl(d(\C,S)\bigr)\leq 2\rho(1+\ve)\OPT$.

To prove the analogue of Claim~\ref{allgood}, suppose all clusters are $\ell$-good for
some $\ell\in\post$.  
%Recall that we have 
Note that $t_\ell\geq{\topl\bigl(d(\C,\optsoln)\bigr)}/{\ell}$ implies that 
$t_\ell\geq t^*_\ell:=d(\C,\optsoln)^{\down}_\ell$.
Similar to the proof of Claim~\ref{allgood}, we have 
\begin{equation*}
\begin{split}
\topl\bigl(d(\C,S)\bigr) & \leq\ell\cdot\beta t_\ell+
\sum_{q=1}^k\sum_{\cli\in \clust}(d(\cli,S)-\beta t_\ell)^+ \\
& \leq 2\beta\cdot\topl\bigl(d(\C,S)\bigr)+
\sum_{q=1}^k\gm\cdot\sum_{\cli\in\clust}(d(\cli,\optcen_q)-t^*_\ell)^+
\leq (2\beta+\gm)\cdot\topl\bigl(d(\C,\optsoln)\bigr).
%\leq\rho\cdot\topl\bigl(d(\C,\optsoln)\bigr).
\end{split}
\end{equation*}
Since $\rho\geq 2\beta+\gm$, we have
$\topl\bigl(d(\C,S)\bigr)\leq\rho\cdot\topl\bigl(d(\C,\optsoln)\bigr)$. 

\medskip
For part (c), we first note that \csalg has the following guarantee: 
given a threshold vector $\vt\in\Rp^{\post}$ and budgets $\{B_\ell\}_{\ell\in\post}$ 
such that $t_\ell\geq d\bigl(\C,\optsoln)^{\down}_\ell$ 
%for all $\ell\in\post$, such that
and $\topl\bigl(d(\C,\optsoln)\bigr)\leq B_\ell$ for all $\ell\in\post$,
\csalg returns a $k$-clustering solution whose $\topl$-cost is bounded by 
$O(\ell t_\ell+B_\ell)$ for all $\ell\in\post$.%
\footnote{See the proof of Theorem 9.2 in~\cite{ChakrabartyS19}.}
%solves the augmented $k$-median LP and rounds its optimal solution. 
%As mentioned earlier, this LP has constraints that bound the $\topl$-proxy cost under the
%vector $\vt$, for all $\ell\in\post$. 
Given the bounds on $t_\ell$, it is valid to set the $\topl$ budgets to $\ell t_\ell$ for
all $\ell\in\post$.
%enforce that the $\topl$-proxy cost is at
%most $\ell t_\ell$, for all $\ell\in\post$. The LP-rounding algorithm in 
Then \csalg returns a solution whose $\topl$-cost is 
$O(\ell t_\ell)=O\bigl(\topl\bigl(d(\C,\optsoln)\bigr)\bigr)$ for all
$\ell\in\post$, 
%then returns an integer solution whose $\topl$-cost is bounded by 
%$O\bigl(\ell t_\ell+\topl\text{-proxy cost of LP solution}\bigr)$%
%for all $\ell\in\post$. Thus, the resulting solution 
which is therefore an $O(1)$-approximate solution.
\end{proof}

\bibliography{references_new}

\end{document}